\documentclass[a4paper,11pt]{article}
\pdfoutput=1 

\usepackage{jheppub} 
\usepackage[bottom]{footmisc}
\usepackage{amsmath,amssymb,amsthm}
\usepackage{extarrows}
\usepackage[dvipsnames]{xcolor}
\usepackage{epsfig}
\usepackage{dcolumn}
\usepackage{tikz}
\usepackage{upgreek}
\usepackage{setspace}
\usepackage{array,multirow,bigdelim,arydshln}
\usepackage{appendix}
\usepackage{xparse}
\usepackage{stmaryrd}
\usepackage[T1]{fontenc}
\usepackage{mathtools}
\usepackage{physics}
\usepackage{graphicx}
\graphicspath{{./images/}}
\usepackage[nottoc]{tocbibind}
\usepackage{hyperref}
\usepackage[utf8]{inputenc}
\usepackage{latexsym}

\usepackage{multicol}
\usepackage{bbm}
\usepackage{enumerate}

\usepackage{adjustbox}
\usepackage[caption=false]{subfig}
\usepackage{mathrsfs}
\usepackage{bm}
\usepackage{scalerel}
\usepackage{diagbox}
\usepackage{tabularx}
\usepackage{CJK}

\usetikzlibrary{shapes.geometric,arrows,arrows.meta,decorations.pathmorphing,decorations.markings,patterns}
\hypersetup{
	colorlinks,
	urlcolor=Maroon,
	linkcolor=Maroon,
	citecolor=Maroon
}

\theoremstyle{plain}
\newtheorem{theorem}{Theorem}[section]

\newtheorem{corollary}[theorem]{Corollary}

\theoremstyle{definition}

\newtheorem*{remark}{Remark}

\def\myTheorem{Claim}
\def\mytheorem{claim}

\font\tenshuffle=shuffle10 \font\sevenshuffle=shuffle7 \font\fiveshuffle=shuffle7 at 5pt
\def\shuffle{{%
    \def\Dshuffle{\mathbin{\hbox{\tenshuffle\char'001}}}%
    \def\Sshuffle{\mathbin{\hbox{\sevenshuffle\char'001}}}%
    \def\SSshuffle{\mathbin{\hbox{\fiveshuffle\char'001}}}%
    \mathchoice{\Dshuffle}{\Dshuffle}{\Sshuffle}{\SSshuffle}}}

\makeatletter
\newcommand{\subalign}[1]{
    \vcenter{%
    \Let@ \restore@math@cr \default@tag
    \baselineskip\fontdimen10 \scriptfont\tw@
    \advance\baselineskip\fontdimen12 \scriptfont\tw@
    \lineskip\thr@@\fontdimen8 \scriptfont\thr@@
    \lineskiplimit\lineskip
    \ialign{\hfil$\m@th\scriptstyle##$&$\m@th\scriptstyle{}##$\hfil\crcr
    #1\crcr
    }%
    }%
}

\newcommand{\raisemath}[1]{\mathpalette{\raisem@th{#1}}}
\newcommand{\raisem@th}[3]{\raisebox{#1}{$#2#3$}}

\def\scalemath#1{\@ifnextchar[{\scalem@th{#1}}{\scalem@th{#1}[#1]}}
\def\scalem@th#1[#2]{\mathpalette{\scalem@@th{#1}[#2]}}
\def\scalem@@th#1[#2]#3#4{\scalebox{#1}[#2]{$#3#4$}}

\newcommand{\vcentermath}{\mathpalette{\vcenterm@th}}
\newcommand{\vcenterm@th}[2]{\vcenter{\hbox{$#1#2$}}}
\makeatother

\newcommand{\eqtitleref}[1]{\texorpdfstring{\eqref{#1}}{(\protect\ref{#1})}}

\SetSymbolFont{stmry}{bold}{U}{stmry}{m}{n}

\usepackage[percent]{overpic}
\usepackage{multirow} 
\usepackage{slashed}

\begin{document}

\begin{CJK*}{UTF8}{}
\CJKfamily{gbsn}

\title{Pions from higher-dimensional gluons: general realizations and stringy models}
\author[a,b]{Jin Dong (董晋),}
\author[a,b]{Xiang Li (李想)}
\author[a,b,c]{and Fan Zhu (朱凡)}

\affiliation[a]{CAS Key Laboratory of Theoretical Physics, Institute of Theoretical Physics, Chinese Academy of Sciences, Beijing 100190, China}
\affiliation[b]{School of Physical Sciences, University of Chinese Academy of Sciences, No.19A Yuquan Road, Beijing 100049, China}
\affiliation[c]{School of Fundamental Physics and Mathematical Sciences, Hangzhou Institute for Advanced Study, UCAS, Hangzhou 310024, China}

\emailAdd{dongjin@itp.ac.cn}
\emailAdd{lixiang@itp.ac.cn}
\emailAdd{zhufan22@mails.ucas.ac.cn}

\abstract{In this paper we revisit the general phenomenon that scattering amplitudes of pions can be obtained from ``dimensional reduction'' of gluons in higher dimensions in a more general context. We show that such ``dimensional reduction'' operations universally turn gluons into pions regardless of details of interactions: under such operations any amplitude that is gauge invariant and contains only local simple poles becomes one that satisfies Adler zero in the soft limit. As two such examples, we show that starting from gluon amplitudes in both superstring and bosonic string theories, the operations produce ``stringy'' completion of pion scattering amplitudes to all orders in $\alpha'$, with leading order given by non-linear sigma model amplitudes. Via Kawai-Lewellen-Tye relations, they give closed-stringy completion for Born-Infeld theory and the special Galileon theory, which are directly related to gravity amplitudes in closed-string theories. We also discuss how they naturally produce stringy models for mixed amplitudes of pions and colored scalars.}

\maketitle

\end{CJK*}

\addtocontents{toc}{\protect\setcounter{tocdepth}{2}}

\numberwithin{equation}{section}

\section{Introduction}
Over the past few decades, enormous progress has been made in new understanding of quantum field theory (QFT) and string theory via the study of the scattering amplitudes. For example, profound relations among scattering amplitudes of gluons, gravitons, and Goldstone particles {\it etc.} have been found, revealing unexpected mathematical structures hidden in these theories. One significant example was found by Kawai-Lewellen-Tye (KLT)~\cite{kawai1986relation} in string theory, unveiling the famous double copy relation between tree-level closed string and open string amplitudes. Modern realization of the double copy in QFT relies on the color-kinematic duality was known as the Bern-Carrasco-Johanssonn (BCJ) relations~\cite{Bern:2008qj,Bern:2010ue} (see~\cite{Bern:2019prr} for a review). Another remarkable foundation is the Cachazo-He-Yuan (CHY) formula~\cite{Cachazo:2013gna,Cachazo:2013iea,Cachazo:2013hca}, which allows us to study a broader range of theories within a unified framework. This includes gauge theories such as Yang-Mills (YM) theory and Yang-Mills-Scalar (YMS) theory, as well as effective field theories (EFTs) such as non-linear sigma model (NLSM), Dirac-Born-Infeld (DBI) theory and special Galileon (sGal) theory.

Based on the CHY formalism, some fascinating relations between gauge theories and EFTs were uncovered in~\cite{Cachazo:2014xea}, revealing that the scattering amplitudes of EFTs are special dimensional reductions (DRs) of gauge theories ones. These relations are further studied in~\cite{Cheung:2017ems,Cheung:2017yef} via fundamental properties of the theories, where the second version of the DRs is derived. In particular, the NLSM amplitudes~\cite{Kampf:2013vha} can be obtained from DRs of Yang-Mills amplitudes, whose two versions of DRs are given by taking the derivative of $A^{\text{YM}}_n$ with respect to Lorentz product of two selected polarization vectors, {\it e.g.} $e_1 \cdot e_2$, then performing replacement~\eqref{eq: DR1and2} with $p_a \cdot p_b$ trivially reduced. The two versions of DRs yield the same result, that is the NLSM amplitude.
\begin{equation} \label{eq: DR1and2}
\begin{aligned}
    \text{DR I}\ \ &:\ e_a \cdot p_b \to 0 ,   \quad e_a \cdot e_b \to -p_a \cdot p_b \\
    \text{DR II}\ &:\ e_a \cdot e_b \to 0 ,  \quad  e_a \cdot p_b \to p_a \cdot p_b
\end{aligned}\quad ,\quad \forall e_a,e_b\notin \{e_1,e_2\}.
\end{equation}

Despite the concision of these relations, their physical significance may raise questions. One may wonder: Are these DRs inherent, or just accidental? Are these relations between gauge theories and EFTs natural, or just contingent upon deliberate selection of the two specific DRs~\eqref{eq: DR1and2} above? In this paper, we will demonstrate the inherent nature of these DR relations between gauge theories and EFTs, by proving that there exists a general class of various DRs~\eqref{eq: DRgeneral-intro} yielding exactly the same result for a gauge invariant object, without any assumption regarding tree-level such as CHY formula and rationality. We further show that the Adler zeros of these resulting EFTs come from locality of the input amplitudes, which reminds us about the intriguing result introduced in~\cite{Arkani-Hamed:2016rak,Rodina:2016jyz} that uniqueness of gauge theory and EFT amplitudes follow from gauge invariance and Adler zeros, respectively.
\begin{equation} \label{eq: DRgeneral-intro}
\hspace{-2em}\text{General DR}\ :\ \;
\begin{aligned}
    &e^{\text{I}}_a\cdot e^{\text{I}}_b \to -p_a \cdot p_b,\ \ 
    e^{\text{I}}_a\cdot e^{\text{II}}_b \to -p_a \cdot p_b,\ \ e^{\text{II}}_a\cdot e^{\text{II}}_b \to 0,\\
    & e^{\text{I}}_a\cdot p_b \to 0,\ \
    e^{\text{II}}_a\cdot p_b \to p_a \cdot p_b ,\ \
    p_a\cdot p_b\to  p_a \cdot p_b.
\end{aligned}
\end{equation}

The generality of the equivalence of DRs inspires us to investigate the EFT amplitudes beyond tree-level, such as their stringy UV completions, which are expected to arise from DRs of the stringy versions of gauge theories.
In this paper, we will mainly focus on the stringy completion of NLSM, for which several different versions of stringy completions have been proposed in~\cite{Carrasco:2016ldy,Carrasco:2016ygv,Bianchi:2020cfc,Arkani-Hamed:2023swr,Arkani-Hamed:2024nhp}. As suggested before, it is natural to identify the DR of ``stringy Yang-Mills'' theory as the stringy NLSM. In order to provide a more concrete manifestation of the equivalence between different DRs, we present a systematic method to establish the equivalence for open superstring via the integral-by-parts (IBP) process~\cite{He:2018pol,He:2019drm}, suggesting that this approach is applicable to any specific string theory like bosonic string. 

We also extend our study to stringy models that concerns the mixed amplitudes of pions and bi-adjoint $\phi^3$ scalars, whose field theory limit has been studied in~\cite{Cachazo:2016njl}. We give a systematic way to compute its low-energy expansion by deriving the ``BCJ numerators'', therefore one can compute the result at any $\alpha^\prime$ order once given the result of Z-integral as computed in~\cite{Mafra:2016mcc}.

At the end of this paper, we extend our IBP-based method to the bosonic string, which is similar to the superstring except that the resulting stringy NLSM may depend on the first derivative $e_i \cdot e_j$ we take beyond leading low energy limit. Besides, it is natural to extend our discussion of the open super and bosonic strings to their closed string version, which immediately shows gravity can be dimensional reduced to Born-Infeld or special Galileon.

This paper is organized as follows: In section~\ref{sec: Dimension reduction: from gluon to pion}, we introduce the general DRs and the corresponding relations between Yang-Mills and NLSM amplitudes. In section~\ref{sec: Adler zero from gauge invariance}, we prove the equivalence of different DRs from gauge invariance, and further the existence of Alder zero after DR from locality. In section~\ref{sec: Pions in open superstrings}, we develop the IBP method to concretely demonstrate the equivalence between DRs for superstring in detail. In section~\ref{sec: Mixed amplitudes and logarithmic forms}, we extend the stringy NLSM to mixed amplitudes of $\phi^3$ and pions scattering and give the logarithmic form. In section~\ref{sec: Pions in open bosonic strings and closed strings}, we apply our IBP method to the bosonic string and apply our result to the closed super and bosonic strings.

\section{Dimension reduction: from gluons to pions}\label{sec: Dimension reduction: from gluon to pion}
Before introducing the general dimensional reductions, let us briefly review the two types of DR relations between tree-level scattering amplitudes of gluons and pions indicated in~\cite{Cachazo:2014xea,Cheung:2017ems}. To illustrate, one should first note that the amplitude of pure pions scattering and one with only two $\phi^3$ scalars are equal, {\it e.g.}
\begin{equation}
    A^\mathrm{NLSM}_n(1,2,\ldots,n)= A^{\mathrm{NLSM}+ \phi^3}_n(1^\phi,2^\phi,3,\ldots,n),
\end{equation}
where we choose the $\phi^3$ scalars to be $1^\phi,2^\phi$, which can be arbitrarily chosen.
This property can be easily understood in the CHY frame work. It is then convenient to start with the YMS amplitude with two scalars, which can be extracted from a pure Yang-Mills one via a differential operator with respect to the two $\phi^3$ scalars:
\begin{equation}
    A^{\mathrm{YM}+ \phi^3}_n(1^\phi,2^\phi,3,\ldots,n)= \partial_{e_1 \cdot e_2} A^\mathrm{YM}_n(1,2,\ldots,n).
\end{equation}

Using the CHY formula, the authors of~\cite{Cachazo:2014xea} found one can suppose that the polarization vector $e_{a}$ and momentum $p_{a}$ of the $a$-th gluon live in dimension $D=2d$, then wisely choose the components of these $2d$-dimensional Lorentz vectors to obtain the NLSM amplitude from Yang-Mills amplitude (or rather, YMS amplitude, after acting the differential operator): 
\begin{equation}\label{eq: DR1}
\text{DR I}\  :\,
\left\{\,\begin{aligned}
    &e_a^M=(0,i p_a^\mu),\quad p_a^M=(p_a^\mu,0)\,,\\
    &e_a\cdot e_b \to - p_a \cdot p_b,\ \  e_a\cdot p_b \to 0, \ \ p_a\cdot p_b\to  p_a \cdot p_b,
\end{aligned}\right.
\end{equation}
where $i$ is the imaginary unit and we use the indices $M$ and $\mu$ for the $2d$- and $d$-dimensional space respectively. It can be shown in the CHY frame that
\begin{equation}
    A^{\mathrm{YM}+ \phi^3}_n(1^\phi,2^\phi,3,\ldots,n) \xrightarrow[]{\eqref{eq: DR1}} A^{\mathrm{NLSM}+ \phi^3}_n(1^\phi,2^\phi,3,\ldots,n).
\end{equation}

Furthermore, one can choose a different DR yielding the same result:
\begin{equation}\label{eq: DR2}
\text{DR II}\  : \,
\left\{\,\begin{aligned}
    &e_a^M=(p_a^\mu,i p_a^\mu),\quad p_a^M=(p_a^\mu,0)\,,\\
    &e_a\cdot e_b \to 0,\ \  e_a\cdot p_b \to p_a \cdot p_b, \ \ p_a\cdot p_b\to  p_a \cdot p_b.
\end{aligned}\right.
\end{equation}

The authors of~\cite{Cheung:2017ems,Cheung:2017yef} provided an explanation in the Lagrangian level that the DR~\eqref{eq: DR2} gives the NLSM amplitude, the proof using CHY formula is also straightforward. 

In fact, the above reductions transform a general mixed amplitude of gluons and $\phi^3$ scalars scattering into pions and $\phi^3$ scalars scattering:
\begin{equation}
    A^{\mathrm{YM}+ \phi^3}_n(\{\bar{\alpha}\}|\alpha) \xrightarrow[]{\eqref{eq: DR1} \text{ or } \eqref{eq: DR2}} A^{\mathrm{NLSM}+ \phi^3}_n(\{\bar{\alpha}\}|\alpha),
\end{equation}
where the implicit overall ordering could be arbitrarily chosen, {\it e.g.} $(1,2,\ldots,n)$. Then the bi-adjoint $\phi^3$ scalars are labeled by an ordered set $\alpha$, with its unordered complementary set $\{\bar{\alpha}\}$ to represent the gluons (pions). Note that for the above reductions to hold we have assumed $2 \leqslant |\alpha| \leqslant n $.

Interestingly, we find that arbitrary combination versions of \eqref{eq: DR1} and \eqref{eq: DR2} do exactly the same transformation for gauge theories. Specifically, to perform a general DR, we split the $2d$-dimensional polarization vectors into two sets $\text{I}$ and $\text{II}$, the two types of dimensional reduced polarization vectors are defined as
\begin{equation}\label{eq: DRgeneral}
\hspace{-1em}\text{General DR}\ :\,
\left\{\,\begin{aligned}
    &e_a^{\text{I}\; M}=(0,i p_a^\mu),\quad e_b^{\text{II}\; M}=(p_b^\mu,i p_b^\mu)\,,\\[1pt]
    &e^{\text{I}}_a\cdot e^{\text{I}}_b \to -p_a \cdot p_b,\ \ 
    e^{\text{I}}_a\cdot e^{\text{II}}_b \to -p_a \cdot p_b,\ \ e^{\text{II}}_a\cdot e^{\text{II}}_b \to 0,\\
    & e^{\text{I}}_a\cdot p_b \to 0,\ \
    e^{\text{II}}_a\cdot p_b \to p_a \cdot p_b ,\ \
    p_a\cdot p_b\to  p_a \cdot p_b.
\end{aligned}\right.
\end{equation}
where we have assumed $a\in \text{I}$ by writing $e_a^{\text{I}\; M}$ and similar for $e_b^{\text{II}\; M}$. For indices in both I and II we have $p_a^{\text{I}\; M}=p_a^{\text{II}\; M}=(p_a^\mu,0)$. For example, for $n=4$ we have:
\begin{equation}
\begin{array}{ll}
    \text{I}=\varnothing\;,\,\text{II}=\{3,4\} &:\ e_3\cdot e_4\to 0,\ \phantom{-} e_3\cdot p_4,e_4\cdot p_3 \to p_3\cdot p_4\\[3pt]
    \text{I}=\{3\}\;,\,\text{II}=\{4\} &:\ e_3\cdot p_4\to 0,\ -e_3\cdot e_4,e_4\cdot p_3 \to p_3\cdot p_4\\[3pt]
    \text{I}=\{3,4\}\;,\,\text{II}=\varnothing &:\ e_3\cdot e_4\to -p_3\cdot p_4,\ e_3\cdot p_4,e_4\cdot p_3 \to 0\\
\end{array}\mspace{2mu}.
\end{equation}

We discover that all the general DRs~\eqref{eq: DRgeneral} yield the same result for a general mixed amplitude of gluons and $\phi^3$, regardless of how we split the particles into I and II:
\begin{equation}
    A^{\mathrm{YM}+ \phi^3}_n(\{\bar{\alpha}\}|\alpha) \xrightarrow[]{\eqref{eq: DRgeneral} } A^{\mathrm{NLSM}+ \phi^3}_n(\{\bar{\alpha}\}|\alpha).
\end{equation}

In the following parts of this paper, we will investigate more deeply into the equivalence of general DRs as well as the resulting pure or mixed stringy NLSM amplitudes, both with non-constructive method for arbitrary gauge theories and constructive method for specific string Yang-Mills theories like the superstring and the bosonic string.

\newpage

\section{Adler zero from gauge invariance}\label{sec: Adler zero from gauge invariance}
As indicated in the introduction, the equivalence among different types of general DRs convinces us of the inherent nature of the DR relations between gauge theories and EFTs. 
In this section, we will demonstrate (1) the equivalence between different types of general DRs~\eqref{eq: DRgeneral} regarding different split of particles into I and II from gauge invariance, and (2) the existence of Adler zero in DR results from locality, as stated in Claim~\ref{thm: equivalent DR} and~\ref{thm: Adler zero} respectively. Throughout this section, $F_n$ is a function of $e_a\cdot e_b,\, e_a\cdot p_b,\, p_a\cdot p_b$ which is multi-linear in $m$ polarization vectors $e_a$ with $m\leqslant n-2$. The on-shell conditions $p^2_a=0$, transversality conditions $e_a \cdot p_a=0$, $e_a \cdot e_a=0$ and momentum conservation $\sum_{a=1}^n p_a=0$ are assumed to hold.

\begin{\mytheorem}\label{thm: equivalent DR}
For arbitrary partition of polarization vectors into sets $\mathrm{I}\mspace{-1mu}$ and $\mspace{1mu}\mathrm{II}$, dimensional reduction~\eqref{eq: DRgeneral} yields the same $F_n^\mathrm{DR}$ for any gauge-invariant $F_n$.
\end{\mytheorem}

\begin{proof}
To manifest the multi-linear structure of $F_n$, it is convenient to decompose $F_n$ into linear independent blocks according to its dependence on $e_a\cdot e_b$, as has been done in~\cite{Pavao:2022kog}.

\begin{equation}\label{eq: Fn decomposition}
    F_n=\sum_{k=0}^{\left \lfloor m/2 \right \rfloor} \sum_{\rho \in \mathfrak{S}_k} \prod_{(a b) \in \rho} e_a \cdot e_b\; F_n^{\rho},
\end{equation}
where $\mathfrak{S}_k$ denotes the set of all possible partitions of $2k$ among $m$ gauge particles into $k$ pairs, each $\rho \in \mathfrak{S}_k$ is an unordered combination of $k$ disjoint non-diagonal pairs of gauge particles. For example, let $n=6$ and the gauge particles be $\{3,4,5,6\}$, we have
\begin{equation}\nonumber
    \mathfrak{S}_1=\{(34), (35), (36), (45), (46), (56)\},\
    \mathfrak{S}_2=\{(34)(56), (35)(46), (36)(45)\},
\end{equation}
and the decomposition of $F_6$ reads
\begin{equation}\nonumber
\begin{aligned}
	F_6&=F_6^{\varnothing} +e_3\mspace{-1.5mu}\cdot\mspace{-1.5mu} e_4 F_6^{(34)}+e_3\mspace{-1.5mu}\cdot\mspace{-1.5mu} e_5 F_6^{(35)}+e_3\mspace{-1.5mu}\cdot\mspace{-1.5mu} e_6 F_6^{(36)}+e_4\mspace{-1.5mu}\cdot\mspace{-1.5mu} e_5 F_6^{(45)}+e_4\mspace{-1.5mu}\cdot\mspace{-1.5mu} e_6 F_6^{(46)}+e_5\mspace{-1.5mu}\cdot\mspace{-1.5mu} e_6 F_6^{(56)}\\
	&+(e_3\cdot e_4) (e_5\cdot e_6) F_6^{(34)(56)}+(e_3\cdot e_5) (e_4\cdot e_6) F_6^{(35)(46)}+(e_3\cdot e_6) (e_4\cdot e_5) F_6^{(36)(45)}.
\end{aligned}
\end{equation}

Now $F_n^{\rho}$ is a function merely of $e_a\cdot p_b,\, p_b\cdot p_c$ with $a\notin \rho$, and is multi-linear in $e_a\cdot p_b$, which helps us to employ the gauge invariance of $F_n$. Recall that the Ward identity requires that applying $e_j\rightarrow p_j$ on $F_n$ for any gauge particle $j$ yields zero. This implies the following condition for a single gauge particle 
\begin{equation}\label{eq: gauge invariance ei}
    F_n^\rho \big|_{e_j \to p_j}= - \sum_{i\in \bar{\rho},i\neq j} e_i \cdot e_j\; F_n^{\rho\mspace{1mu} \sqcup\mspace{1mu} (ij)}\,\Big|_{e_j\rightarrow p_j}\,,\quad \forall \rho\cap\{j\}=\varnothing.
\end{equation}

For the above example of $F_6$ with $e_6\to p_6$, this just comes from
\begin{equation}\nonumber
\begin{aligned}
	0&=\vcentermath{\scalemath{1}[1.5]{\{}} \big[F_6^{\varnothing}+e_3\cdot p_6 F_{6}^{(36)}+e_4\cdot p_6 F_{6}^{(46)}+e_5\cdot p_6 F_{6}^{(56)} \big] +e_3\cdot e_4 \big[ F_{6}^{(34)}+e_5\cdot p_6 F_{6}^{(34)(56)} \big]\\
	&+e_3\cdot e_5 \big[ F_{6}^{(35)}+e_4\cdot p_6 F_{6}^{(35)(46)} \big] +e_4\cdot e_5 \big[ F_{6}^{(45)}+e_3\cdot p_6 F_{6}^{(36)(45)} \big]\vcentermath{\scalemath{1}[1.5]{\}}}\mspace{1mu}\Big|_{e_6\to p_6},
\end{aligned}
\end{equation}
which implies that all terms in the square brackets must be zero.

However, since DR~\eqref{eq: DRgeneral} involves replacement on not only a single gauge particle, we need to generalize~\eqref{eq: gauge invariance ei} into one with replacement on a nonempty set $\mathfrak{I}$ of gauge particles $e^\mathfrak{I}\rightarrow p^\mathfrak{I}$. Note that in order to get a useful gauge invariance condition with respect to $\mathfrak{I}$ in our proof, we should not naively consider the Ward identity with $e^\mathfrak{I}\rightarrow p^\mathfrak{I}$ acting on $F_n$, but should apply~\eqref{eq: gauge invariance ei} recursively to get:
\begin{equation}\label{eq: gauge invariant eI}
    F_n^\rho \big|_{e^\mathfrak{I} \to p^\mathfrak{I}}=\sum_{k=0}^{\left \lfloor m/2 \right \rfloor} \sum_{\sigma\in(\mathfrak{S}_k|\rho,\mathfrak{I})} (-1)^{k}\prod_{(a b) \in \sigma} e_a \cdot e_b\; F_n^{\rho\mspace{1mu} \sqcup\mspace{1mu} \sigma}\,\Big|_{e^\mathfrak{I} \to p^\mathfrak{I}}\,,\quad \forall \rho\cap\mathfrak{I}=\varnothing,
\end{equation}
where $(\mathfrak{S}_k|\rho,\mathfrak{I})$ denotes the set of all $\sigma\in\mathfrak{S}_k$ such that $\mathfrak{I}\subset\sigma\subset\bar{\rho}$ and each pair in $\sigma$ has nonempty intersection with $\mathfrak{I}$. As a consequence, for each pair $(a b) \in \sigma$, at least one of the polarization vectors in $e_a \cdot e_b$ would be replaced by $e^\mathfrak{I} \to p^\mathfrak{I}$, hence the RHS depends merely on $e_a\cdot p_b,\, p_b\cdot p_c$ and no $e_a \cdot e_b$ will appear.

Before proving~\eqref{eq: gauge invariant eI}, let us illustrate the notation introduced here with two examples. For the above example of $F_6$ with $\mathfrak{I}=\{5,6\}$, we have
\begin{equation}\nonumber
    (\mathfrak{S}_1|\varnothing,\mathfrak{I})=(\mathfrak{S}_1|(34),\mathfrak{I})=\{(56)\},\ (\mathfrak{S}_2|\varnothing,\mathfrak{I})=\{(35)(46),(36)(45)\},\ (\mathfrak{S}_2|(34),\mathfrak{I})=\varnothing.
\end{equation}

Since the $n=6$ example seems a bit unrepresentative, we also give a more typical but incomplete example with $n=8$ and gauge particles being $\{3,4,5,6,7,8\}$, $\mathfrak{I}=\{7,8\}$:
\begin{equation}\nonumber
\begin{aligned}
    &(\mathfrak{S}_2|\varnothing,\mathfrak{I})=\{(37)(48),\dots,(58)(67)\}=\{(i7)(j8),(j7)(i8)|(ij)\in\mathfrak{S}_1,(ij)\cap\mathfrak{I}=\varnothing\},\\
    &(\mathfrak{S}_2|(34),\mathfrak{I})=\{(57)(68),(58)(67)\},\ (\mathfrak{S}_2|(34)(56),\mathfrak{I})=\varnothing,\ \cdots
\end{aligned}
\end{equation}

Gauge invariance condition~\eqref{eq: gauge invariant eI} can be proved by induction on $|\mathfrak{I}|$. For $|\mathfrak{I}|=1$, it is just~\eqref{eq: gauge invariance ei}. Assume~\eqref{eq: gauge invariant eI} holds for $|\mathfrak{I}|=k$, then $\forall\rho\cap(\mathfrak{I}\sqcup\{j\})=\varnothing$ we have
\begin{equation}\nonumber
\begin{aligned}
    F_n^\rho \big|_{e^{\mathfrak{I}\sqcup\{j\}} \to p^{\mathfrak{I}\sqcup\{j\}}}
    &= \sum_{k=0}^{\left \lfloor m/2 \right \rfloor} \sum_{\sigma\in(\mathfrak{S}_k|\rho,\mathfrak{I})} (-1)^{k}\prod_{(a b) \in \sigma} e_a \cdot e_b\; F_n^{\rho\mspace{1mu} \sqcup\mspace{1mu} \sigma}\,\Big|_{e^{\mathfrak{I}\sqcup\{j\}} \to p^{\mathfrak{I}\sqcup\{j\}}}\\
    &= \sum_{k=0}^{\left \lfloor m/2 \right \rfloor} \sum_{
    \begin{smallmatrix*}[c]
    \sigma \in (\mathfrak{S}_k|\rho,\mathfrak{I}) \\
    (ij) \in (\mathfrak{S}_1|\rho,\mathfrak{I}) 
    \end{smallmatrix*}} (-1)^{k+1} \prod_{(a b) \in \sigma\mspace{1mu} \sqcup\mspace{1mu}(ij)} e_a \cdot e_b\; F_n^{\rho\mspace{1mu} \sqcup\mspace{1mu} \sigma\mspace{1mu} \sqcup\mspace{1mu}(ij)}\,\Big|_{e^{\mathfrak{I}\sqcup\{j\}} \to p^{\mathfrak{I}\sqcup\{j\}}}\\
    &= \sum_{k=0}^{\left \lfloor m/2 \right \rfloor} \sum_{\sigma\mspace{1mu} \sqcup\mspace{1mu}(ij)\in (\mathfrak{S}_{k+1}|\rho,\mathfrak{I})} (-1)^{k+1}\prod_{(a b) \in \sigma\mspace{1mu} \sqcup\mspace{1mu}(ij)} e_a \cdot e_b\; F_n^{\rho\mspace{1mu} \sqcup\mspace{1mu} \sigma\mspace{1mu} \sqcup\mspace{1mu}(ij)}\,\Big|_{e^{\mathfrak{I}\sqcup\{j\}} \to p^{\mathfrak{I}\sqcup\{j\}}},
\end{aligned}
\end{equation}
where the second equality comes from~\eqref{eq: gauge invariance ei} and the third equality is due to the fact that $\mathfrak{I}\subset\sigma\sqcup(ij)\subset\bar{\rho}$ if $\mathfrak{I}\subset\sigma\subset\bar{\rho}$ and $\mathfrak{I}\subset(ij)\subset\bar{\rho}$. Hence~\eqref{eq: gauge invariant eI} holds for any $\mathfrak{I}$.

Let us illustrate it with the above $n=6$ example. From~\eqref{eq: gauge invariance ei} we have
\begin{equation}\nonumber
\begin{aligned}
    &F_6^{\varnothing}\big|_{e_6\to p_6}=-e_3\cdot p_6 F_{6}^{(36)}-e_4\cdot p_6 F_{6}^{(46)}-e_5\cdot p_6 F_{6}^{(56)}\Big|_{e_6\to p_6}\,,\\[3pt] 
    &F_{6}^{(36)}\big|_{e_5\to p_5}=-e_4\cdot p_5 F_{6}^{(36)(45)}\Big|_{e_5\to p_5}\,,\ \ F_{6}^{(46)}\big|_{e_5\to p_5}=-e_3\cdot p_5 F_{6}^{(35)(46)}\Big|_{e_5\to p_5}\,. 
\end{aligned}
\end{equation}

Then we perform induction to yield the identity for $\rho=\varnothing$, $\mathfrak{I}=\{5,6\}$. By substituting the above identities of $\mathfrak{I}=\{5\}$ into that of $\mathfrak{I}=\{6\}$, we derive the desired result
\begin{equation}\nonumber
    F_6^{\varnothing}\Big|_{\raisemath{5pt}{\substack{e_5\to p_5\\ e_6\to p_6}}}=-p_5\cdot p_6 F_{6}^{(56)}+(e_3\cdot p_5)(e_4\cdot p_6) F_{6}^{(35)(46)}+(e_3\cdot p_6)(e_4\cdot p_5) F_{6}^{(36)(45)}\Big|_{\raisemath{5pt}{\substack{e_5\to p_5\\ e_6\to p_6}}}\,.
\end{equation}

Now we are just one step away from the conclusion. Let $\rho=\varnothing$, then $\rho\cap\mathfrak{I}=\varnothing$ and $\sigma\subset\bar{\rho}$ are automatically satisfied, so we can abbreviate $(\mathfrak{S}_k|\rho,\mathfrak{I})$ to $(\mathfrak{S}_k|\mathfrak{I})$. By further taking $e_i\to p_i$ for other $i\in\bar{\mathfrak{I}}$ as well, we have
\begin{equation}\label{eq: equivalent DR proof}
    F_n^\varnothing \big|_{e \to p}=\sum_{k=0}^{\left \lfloor m/2 \right \rfloor}  \sum_{\sigma\in (\mathfrak{S}_k|\mathfrak{I})}(-1)^{k}\prod_{(a b) \in \sigma} p_a \cdot p_b\; F_n^{\sigma}\,\Big|_{e \to p}.
\end{equation}

So far we have already completed the proof. Note that the LHS of~\eqref{eq: equivalent DR proof} is nothing but the reduction~\eqref{eq: DR2} acting on $F_n$, or equivalently~\eqref{eq: DRgeneral} with $\mathrm{I}=\varnothing$:
\begin{equation}
   F_n \xrightarrow{\, \eqref{eq: DRgeneral} \text{ with } \mathrm{I}=\varnothing \;} \left. F_n^\varnothing \right|_{e \to p}.
\end{equation}

And the RHS of~\eqref{eq: equivalent DR proof} is exactly DR~\eqref{eq: DRgeneral} with $\mathrm{I}=\mathfrak{I}\neq\varnothing$ acting on $F_n\,$:
\begin{equation}
   F_n \xrightarrow{\, \eqref{eq: DRgeneral} \text{ with } \mathrm{I}=\mathfrak{I} \;}\sum_{k=0}^{\left \lfloor m/2 \right \rfloor}  \sum_{\rho\in (\mathfrak{S}_k|\mathfrak{I})} (-1)^{k}\prod_{(a b) \in \rho} p_a \cdot p_b\; F_n^{\rho}\,\Big|_{e \to p}.
\end{equation}

Hence all the $F_n^\mathrm{DR}$ are equal. This completes the proof.
\end{proof}

\begin{corollary}\label{cor: oddDRvanish}
    For gauge-invariant $F_n$ with odd $m$, DR~\eqref{eq: DRgeneral} yields zero.
\end{corollary}

\begin{proof}
Let $\bar{\mathrm{I}}=\{j\}$, then $(\mathfrak{S}_k|\mspace{1mu}\mathrm{I})\neq\varnothing$ only if $k=\frac{m-1}{2}$. While for $\rho\in(\mathfrak{S}_{(m-1)/2}|\mspace{1mu}\mathrm{I})$, one can observe that $\bar{\rho}=\{j\}$ and the only polarization vector that $F_n^\rho$ depends on is $e_j$. This allows us to solely replace $e_j\to p_j$ to obtain the reduce $F_n$ under DR~\eqref{eq: DRgeneral} as follows:
\begin{equation}\label{eq: odd m DR}
   F_n \xrightarrow{\, \eqref{eq: DRgeneral} \text{ with } \bar{\mathrm{I}}=\{j\} \;} (-1)^{(m-1)/2} \sum_{\rho\in (\mathfrak{S}_{(m-1)/2}|\mspace{1mu}\mathrm{I})}\prod_{(a b) \in \rho} p_a \cdot p_b\; F_n^{\rho}\big|_{e_j \to p_j}.
\end{equation}

According to~\eqref{eq: gauge invariance ei}, this must be zero since there exists no non-diagonal $(ij)\subset\bar{\rho}$.
\end{proof}

\begin{\mytheorem}\label{thm: Adler zero}
    The resulting $F_n^\mathrm{DR}$ has Adler zero if $F_n$ only contains local simple poles.
\end{\mytheorem}

\begin{proof}
Let us first consider the soft behaviour of $F_n$ before dimensional reduction, where we set $p_n = t\, \hat{p}_n$ with $t \to 0$. Denote the $O(t^k)$ order of Laurent series of $F_n$ by $\mathcal{F}^{k}_n$, then generally the leading order $\mathcal{F}^{-1}_n$ takes the following form:

\addtolength{\belowdisplayskip}{-0.3\baselineskip}

\begin{equation}
    \mathcal{F}^{-1}_n=  \sum_{i,j\neq n} \frac{e_n \cdot p_j}{p_n \cdot p_i} B^{(0)}_{ij} + \frac{e_n \cdot e_j}{p_n \cdot p_i} C^{(0)}_{ij},\quad B^{(0)}_{ij},C^{(0)}_{ij}\sim O(t^0).
\end{equation}

Note that (weak) locality condition forbids any pole in form of $(\hat{p}_n\cdot p_k)^{-1}$ in $B^{(0)}_{ij},C^{(0)}_{ij}$,\pagebreak[4] hence $B^{(0)}_{ij},C^{(0)}_{ij}$ are free of $\hat{p}_n$. Gauge invariance further demands that

\addtolength{\belowdisplayskip}{0.3\baselineskip}

\begin{equation}
    \mathcal{F}^{-1}_n\big|_{e_n \to p_n}=  \sum_{i,j\neq n} \frac{p_n \cdot p_j}{p_n \cdot p_i} B_{ij}^{(0)} + \frac{p_n \cdot e_j}{p_n \cdot p_i} C_{ij}^{(0)}=0.
\end{equation}

Note that $(p_n \cdot p_j)/(p_n \cdot p_i)$ with $i\neq j$ and all $(p_n \cdot e_j)/(p_n \cdot p_i)$ are linear independent, we must have $B^{(0)}_{ij}=0$ if $i\neq j$ and $C^{(0)}_{ij}=0$. This implies the \textit{soft theorem}

\begin{equation}
    \mathcal{F}^{-1}_n=  \sum_{i\neq n} \frac{e_n \cdot p_i}{p_n \cdot p_i} B^{(0)}_i,\;\ \text{with}\;\ \sum_{i\neq n} B^{(0)}_i=0.
\end{equation}

\addtolength{\belowdisplayskip}{-0.3\baselineskip}

Let $\mathrm{I}=\{n\}$, obviously $\mathcal{F}^{-1}_n$ under DR is zero:
\begin{equation}
    \mathcal{F}^{-1}_n\big|_{\text{DR with I}=\{n\}}=\sum_{i\neq n} \frac{e_n \cdot p_i}{p_n \cdot p_i} B^{(0)}_i\Big|_{ e_n\cdot p_i \to 0,\, \cdots}=0.
\end{equation}

\addtolength{\belowdisplayskip}{0.1\baselineskip}

On the other hand, the next-leading order $\mathcal{F}^{0}_n$ takes the following form:
\begin{equation}
    \mathcal{F}^{\,0}_n = \sum_{i,j\neq n} \frac{e_n \cdot p_j}{p_n \cdot p_i} B^{(1)}_{ij} + \frac{e_n \cdot e_j}{p_n \cdot p_i} C^{(1)}_{ij},\quad B^{(1)}_{ij},C^{(1)}_{ij}\sim O(t^1).
\end{equation}

Similarly, take $\mathrm{I}=\{n\}$, it is easy to see that
\begin{equation}
    \mathcal{F}^{\,0}_n\big|_{\text{DR with I}=\{n\}}=-\sum_{i,j\neq n} \frac{p_n \cdot p_j}{p_n \cdot p_i} C^{(1)}_{ij}\Big|_{\cdots}\sim O(t^1).
\end{equation}

\addtolength{\belowdisplayskip}{0.2\baselineskip}

Therefore $F_n^\mathrm{DR}$ must be at least of order $O(t^1)$. This completes the proof.
\end{proof}

\begin{remark}
Earlier literature~\cite{Arkani-Hamed:2016rak,Rodina:2016jyz} has found that assuming locality and the correct power-counting, uniqueness of (1) Yang-Mills and gravity, (2) NLSM and DBI tree amplitudes is ensured by (1) gauge invariance and (2) (enhanced) Adler zeros. Our {\mytheorem}s reveal that the two conditions are not independent, but rather related through dimensional reduction.
\end{remark}

\section{Pions in open superstrings}\label{sec: Pions in open superstrings}
Now we have obtained a non-constructive proof on the equivalence between different types of general DRs. However, we will not be content with just knowing \textit{that} those DRs are equal, but would also wonder \textit{how} they are equal for a specific string theory. Due to the extreme difficulty of evaluating string amplitudes, we develop a systematic method based on stringy IBP~\cite{He:2018pol,He:2019drm} to verify the DR relations at the integrand level. 
In this section, we first review the superstring integral and introduce a set of functions for $m$ gauge particles called \textit{IBP building blocks}~\eqref{def: IBP building blocks}. Next, we demonstrate the equivalence of all IBP building blocks within a family via certain IBP processes~\eqref{eq: IBP of IBP buildingblocks}. Finally, we set $m=n-2$ to show that the dimensional reduced $n$-particle superstring correlator can be written as linear combination of IBP building blocks~\eqref{eq: expand DR superstring correlator}, therefore establishing the equivalence of different DRs for superstring amplitudes under IBP.
We give a specific example for $n=5$ points at the end of this section and put off most of the tedious algebraic operations to Appendix~\ref{app.: Supplement for superstring theory}.
	
\subsection{Open superstrings and IBP building blocks}
The generic massless open-string tree amplitude is given by a disk integral:
\addtolength{\belowdisplayskip}{-0.2\baselineskip}
\begin{equation}
    \mathcal{M}_n^{\text {string }}(\rho)=\int_\rho \underbrace{\frac{d^n z}{\operatorname{vol~SL}(2, \mathbb{R})} \prod_{i<j}\left|z_{i j}\right|^{\alpha's_{i j}}}_{:=d \mu_n^{\text {string }}} \mathcal{I}_n^{\text {string }}(z), \quad \mathrm{KN}:=\prod_{i<j}\left|z_{i j}\right|^{\alpha's_{i j}}
\end{equation}
where $z_{ij}:=z_i - z_j$ and $s_{ij}:=2p_i\cdot  p_j$ are the Mandelstam variables. The color ordering $\rho\in S_n /\mathbb{Z}_n$ is realized by the integration domain $z_{\rho(i)} < z_{\rho(i+1)}$. We denote the Koba-Nielsen factor as KN and the integral measure including it as $d \mu_n^{\text{string}}$. 
Using the $\text{SL}(2,\mathbb{R})$ redundancy one can fix e.g. $(z_1, z_{n-1}, z_n)=(0, 1, \infty)$, and the product in the Koba-Nielsen factor goes over all $i,j$ such that $1\leqslant i < j\leqslant n-1$ with this fixing. The string correlator $\mathcal{I}_n^{\text{string}}$ is a rational function of $z$'s, which is required to have the correct SL(2) weight: $\mathcal{I}_n^{\text{string}}\to\prod_{a=1}^n (\gamma+\delta z_a)^2\mathcal{I}_n^{\text{string}}$ under $z_a\to-\frac{\alpha+\beta z_a}{\gamma+\delta z_a}$ with $\alpha\delta-\beta\gamma=1$. For convenience in the following text, let us introduce the $2n\times 2n$ skew-symmetric matrix $\mathbf\Psi$ constructed by:
\addtolength{\abovedisplayskip}{-0.25\baselineskip}
\begin{equation}\label{def: CHY building blocks}
\begin{aligned}
&({\mathbf A})_{ij} := \begin{dcases}
	\phantom{\frac{1}{1}}\ 0 & \text{if}\quad i = j,\\[-2pt]
	\frac{2p_{i} \cdot p_{j}}{z_{ij}} & \text{otherwise},
\end{dcases} \qquad
({\mathbf B})_{ij} := \begin{dcases}
	\phantom{\frac{1}{1}}\ 0 & \text{if}\quad i = j,\\[-2pt]
	\frac{2 e_i \cdot  e_j}{z_{ij}} & \text{otherwise},
\end{dcases} \\[5pt]
&({\mathbf C})_{ij} := \begin{dcases}
	-\sum_{k\neq i} \frac{2 e_i \cdot p_k}{z_{ik}} & \text{if}\quad i = j,\\
	\ \ \frac{2 e_i \cdot p_j}{z_{ij}} & \text{otherwise},
\end{dcases}\qquad 
	{\mathbf\Psi} := \left(\begin{array}{cc}
	\mathbf{A} & -{\mathbf C}^\text{T}\\
	{\mathbf C} & \mathbf{B} \\
    \end{array}\right)
\end{aligned}
\end{equation}

As shown in~\cite{Mizera:2019gea}, the right/left-moving superstring integrand reads:
\begin{equation}\label{eq: superstring correlator}
    \varphi_{\pm,n}^{\text{gauge}}=\frac{1}{z_{i_0 j_0}}\sum_{q=0}^{\lfloor n/2\rfloor -1}(\mp\alpha')^{-q}\sum_{\rho\in\mathfrak{S}_q}\prod_{(ij)\in\rho}\frac{2 e_{i}\cdot e_{j}}{z_{i j}^2}\;\operatorname{Pf}{\mathbf\Psi}_{i_0,j_0}^{\rho}
\end{equation}
where $\pm$ denotes the right/left movers respectively, and we only consider $\varphi_{+,n}^{\text{gauge}}$ (or $\varphi_{n}^{\text{gauge}}$ in short) here since we will focus on the open string, leaving the discussion of closed string to Section~\ref{sec: Pions in open bosonic strings and closed strings}. In this section we choose $(i_0,j_0)=(1,2)$, which can be chosen arbitrarily, since different choices of $(i_0,j_0)$ are cohomologous to each other. The second sum goes over all possible partitions of $2q$ particles among $\{3,\ldots,n\}$ into $q$ distinct pairs, as defined in Section~\ref{sec: Adler zero from gauge invariance}. Here we use ${\mathbf\Psi}^{\rho}_{12}$ to denote the $2(n{-}2q{-}1)\times 2(n{-}2q{-}1)$ matrix yielded by removing the 1st, 2nd and the $i,j,n{+}i,n{+}j$-th columns and rows from $\mathbf\Psi$ for each $(ij)\in\rho$.
\addtolength{\abovedisplayskip}{0.25\baselineskip}
\addtolength{\belowdisplayskip}{0.2\baselineskip}

To enhance understanding of the notation, let us give a few leading terms in~\eqref{eq: superstring correlator}:
\begin{align*}
\varphi^{\text{gauge}}_{n}&=\frac{1}{z_{12}} \Bigg[\, \text{Pf}\ {\mathbf\Psi}_{12} -\frac{1}{\alpha'} \sum_{3\leqslant i_1 < j_1 \leqslant n} \frac{2 e_{i_1} \cdot  e_{j_1}}{z_{i_1 j_1}^2} \text{Pf}\ {\mathbf\Psi}^{i_1 j_1}_{12} \\
&+\frac{1}{\alpha'^2}\Bigg(\sum_{3 \leqslant i_1 < j_1 < i_2 < j_2 \leqslant n} \!{+}\! \sum_{3 \leqslant i_1 < i_2 < j_1 < j_2 \leqslant n}\Bigg) \frac{2 e_{i_1} \cdot  e_{j_1}}{z_{i_1j_1}^2} \frac{2 e_{i_2} \cdot  e_{j_2}}{z_{i_2j_2}^2} \text{Pf}\ {\mathbf\Psi}^{i_1 j_1 i_2 j_2}_{12} + \cdots \Bigg]
\end{align*}

Directly performing IBP on integrand~\eqref{eq: superstring correlator} would be tedious and lack of universality. This motivates us to consider the IBP-equivalent function families, namely the IBP building blocks, which consist of the two letters:
\begin{equation}\label{def: VW}
V_i:=\sum _{j\neq i} \frac{s_{i j}}{z_i-z_j},\qquad W_{ij}:=-\frac{s_{ij}}{\alpha'(z_i-z_j)^2}.
\end{equation}

Importantly, the two letters and the Koba-Nielsen factor are related by:
\begin{equation}\label{derivative of VW}
    \partial_i \operatorname{KN}=\alpha'V_i\cdot\operatorname{KN},\quad \partial_i V_j=-\alpha'W_{ij},\qquad \partial_i:=\partial_{z_i}
\end{equation}

\addtolength{\abovedisplayskip}{-0.2\baselineskip}

For a given split of $m$ gauge particles into $\text{I}\sqcup\text{II}$, the IBP building block is defined as:
\begin{equation}\label{def: IBP building blocks}
    \mathcal{I}({\text{I}|\text{II}})=\sum_{q=0}^{\lfloor m/2 \rfloor} \sum_{\rho\in(\mathfrak{S}_q|\text{I})}\prod_{(ij)\in \rho}W_{ij} \; \prod_{k\in\bar{\rho}}V_{k}
\end{equation}
where $\bar{\rho}$ is the complement of $\rho$ with respect to $\text{I}\sqcup\text{II}$, $(\mathfrak{S}_q|\text{I})$ denotes the set of all $\rho\in\mathfrak{S}_q$ such that $\text{I}\subset\rho$ and each pair in $\rho$ has nonempty intersection with $\text{I}$, as defined before~\eqref{eq: equivalent DR proof}. Here are some examples and relations for IBP building blocks:
\addtolength{\abovedisplayskip}{0.2\baselineskip}
\begin{equation}\label{eq: examples of IBP building blocks}
\begin{aligned}
&\mathcal{I}(\{1,3\}|\{2,4\})=V_2 V_4 W_{13}+W_{14} W_{23}+W_{12} W_{34}\\[7pt]
&\mathcal{I}(\{1,2,4\}|\{3,5\})=V_5 W_{1,4} W_{2,3}+V_5 W_{1,3} W_{2,4}+V_3 W_{1,5} W_{2,4}+V_3 W_{1,4} W_{2,5}\\
&\qquad\qquad\qquad\qquad\ +V_5 W_{1,2} W_{3,4}+V_3 W_{1,2} W_{4,5}\\[7pt]
&\mathcal{I}(\{1,2,\dots,2n-1\}|\varnothing)=0\qquad 
\mathcal{I}(\varnothing|\{1,2,\dots,n\})=V_1 V_2\cdots V_n\\[7pt]
&\mathcal{I}(\{1,2,\dots,2n-1\}|\{2n\})=\mathcal{I}(\{1,2,\dots,2n\}|\varnothing)\\[7pt]
&\mathcal{I}(\{1,2,\dots,2n-2\}|\{2n-1\})=V_{2n-1}\mathcal{I}(\{1,2,\dots,2n-2\}|\varnothing)\\[2pt]
\end{aligned}
\end{equation}

We define a family of IBP building blocks as those that share identical $\text{I}\sqcup\text{II}$, these IBP building block families turn out to be exactly the maximal IBP-equivalent function families we need. Proof of the equivalence will be given in the next subsection.

\begin{remark}
    From~\eqref{eq: examples of IBP building blocks} we see that $\mathcal{I}(\{1,2,\dots,n\}|\varnothing){=}0$ for any odd $n$, which is consistent with Corollary~\ref{cor: oddDRvanish} and the known property of vanishing odd-point amplitudes in NLSM.
\end{remark}

\subsection{The equivalence of DRs under IBP}
\subsubsection{The equivalence of~\texorpdfstring{$\mathcal{I}(\text{I}|\text{II})$}{I(I|II)}~under IBP}
In the forthcoming discussion, we consider a $n$-point scattering process with the particles ordered as $\{1,2,\dots,n\}$. For two IBP building blocks within the same family of the form $\mathcal{I}(\{i_1,\ldots,i_{k}\}|\{i_{k+1},\ldots,i_m\})$ and $\mathcal{I}(\{i_1,\ldots,i_{k+1}\}|\{i_{k+2},\ldots,i_m\})$, once we prove that such two IBP building blocks are IBP-equivalent, it follows from transitivity that all IBP building blocks in the same family are IBP-equivalent. For simplicity, let us temporarily abbreviate $\mathcal{I}(\{i_1,\ldots,i_{k}\}|\{i_{k+1},\ldots,i_m\})$ to $\mathcal{I}(\text{I}'|\text{II}')$ and $\mathcal{I}(\{i_1,\ldots,i_{k+1}\}|\{i_{k+2},\ldots,i_m\})$ to $\mathcal{I}(\text{I}''|\text{II}'')$. To prove the IBP-equivalence between $\mathcal{I}(\text{I}'|\text{II}')$ and $\mathcal{I}(\text{I}''|\text{II}'')$, it is necessary to look into their summation ranges for a given $q$, namely $(\mathfrak{S}_q|\mathrm{I}')$ and $(\mathfrak{S}_q|\mathrm{II}')$, and tell their difference. With careful observation, we find that
\begin{equation}
\begin{aligned}
    (\mathfrak{S}_q|\mathrm{I}')\backslash (\mathfrak{S}_q|\mathrm{I}'') &= \left\{\rho\in(\mathfrak{S}_{q}|\mathrm{I}')\mspace{1mu}\big{|}\mspace{2mu}i_{k+1}\in\bar{\rho} \mspace{1mu}\right\},\\
    (\mathfrak{S}_q|\mathrm{I}'')\backslash (\mathfrak{S}_q|\mathrm{I}') &= \left\{\rho\in(\mathfrak{S}_{q}|\mathrm{I}')\mspace{1mu}\big{|} \mspace{1mu}(i_{k+1}i_\ell)\in\rho\mspace{2mu},\mspace{1mu}i_\ell\in\mathrm{II}''\right\}.
\end{aligned}
\end{equation}

That is to say, the terms in $\mathcal{I}(\text{I}'|\text{II}')$ but not in $\mathcal{I}(\text{I}''|\text{II}'')$ are those containing $V_{i_{k+1}}$, while those in $\mathcal{I}(\text{I}''|\text{II}'')$ but not in $\mathcal{I}(\text{I}'|\text{II}')$ are those containing $W_{i_{k+1}i_\ell}$ for some $i_\ell\in\mathrm{II}''$. This inspires us to extract the respective common factors $V_{{i_{k+1}}}$ and $W_{{i_{k+1}i_\ell}}$. Note from~\eqref{def: IBP building blocks} that the coefficient of $V_i$ or $W_{ij}$ in $\mathcal{I}(\text{I}|\text{II})$ is a minor IBP building block excluding $i$ or $i,j$ from the ``gauge particles'' list $\text{I}\sqcup\text{II}$, we can express $\mathcal{I}(\text{I}'|\text{II}')-\mathcal{I}(\text{I}''|\text{II}'')$ in the minor IBP building blocks $\mathcal{I}(\text{I}'|\text{II}'')$ and $\mathcal{I}(\text{I}'|\text{II}''\mspace{1mu}\backslash\{i_\ell\})$ as follows:

\begin{equation}\label{eq: IBP building blocks difference}
    \mathcal{I}(\text{I}'|\text{II}')-\mathcal{I}(\text{I}''|\text{II}'')=V_{i_{k+1}}\mathcal{I}(\text{I}'|\text{II}'')-\sum_{i_\ell\in\text{II}''} W_{i_{k+1}i_\ell}\mspace{2mu} \mathcal{I}(\text{I}'|\text{II}''\mspace{1mu}\backslash\{i_\ell\})\mspace{2mu}.
\end{equation}

Then we will observe from IBP and the Leibniz rule that the two IBP building blocks $\mathcal{I}(\text{I}'|\text{II}')$ and $\mathcal{I}(\text{I}''|\text{II}'')$ multiplied by KN are equivalent as expected:
\begin{equation}\label{eq: IBP of IBP buildingblocks}
\begin{aligned}
    \text{KN}\cdot V_{i_{k+1}}\mathcal{I}(\text{I}'|\text{II}'')=\frac{1}{\alpha'}(\partial_{i_{k+1}}\text{KN})\cdot\mathcal{I}(\text{I}'|\text{II}'')&\xlongequal{\text{IBP}}-\frac{1}{\alpha'}\text{KN}\cdot\partial_{i_{k+1}}\mathcal{I}(\text{I}'|\text{II}'')\\[4pt]
    =-\frac{1}{\alpha'}\text{KN}\cdot\sum_{i_\ell\in\text{II}''} (\partial_{i_{k+1}}V_{i_\ell})\mspace{2mu} \mathcal{I}(\text{I}'|\text{II}''\backslash\{i_\ell\})&=\text{KN}\cdot\sum_{i_\ell\in\text{II}''} W_{i_{k+1}i_\ell}\mspace{2mu} \mathcal{I}(\text{I}'|\text{II}''\backslash\{i_\ell\})\mspace{2mu},\\
    \Rightarrow \ \text{KN}\cdot \mathcal{I}(\text{I}'|\text{II}') \xlongequal{\text{IBP}} &\; \text{KN}\cdot \mathcal{I}(\text{I}''|\text{II}'')\mspace{2mu}.
\end{aligned}
\end{equation}

The boundary term emerging from IBP vanishes due to the short-distance behavior of KN, {\it i.e.} $\operatorname{KN}|_{z_i\to z_j}\to 0$. Now starting from a given IBP building block, we can transitively prove its equivalence to any other IBP building blocks in the same family with a sequence of IBPs and appropriately chosen $\text{SL}(2,\mathbb{R})$ gauge fixing such that $z_{i_{k+1}}$ is free coordinate at each step. From the preceding argument we have proved that

\begin{\mytheorem}\label{thm: IBP building block equivalence}
For IBP building blocks in the same family we have
\begin{equation}
    \int_\rho d\mu_n^\mathrm{string}~\mathcal{I}(\mathrm{I}_1|\mathrm{II}_1)\xlongequal{\mathrm{IBP}}\int_\rho d\mu_n^\mathrm{string}~\mathcal{I}(\mathrm{I}_2|\mathrm{II}_2)\quad,\quad\forall\,\mathrm{I}_1\sqcup\mathrm{II}_1=\mathrm{I}_2\sqcup\mathrm{II}_2
\end{equation}
\end{\mytheorem}

To end this subsection, let us illustrate \myTheorem~\ref{thm: IBP building block equivalence} with the following example:
\begin{equation}\nonumber
\begin{aligned}
    \mathcal{I}(\{1,3\}|\{2,4\})-\mathcal{I}(\{1,2,3\}|\{4\})&=(V_2\mathcal{I}(\{1,3\}|\{4\})+W_{14} W_{23}+W_{12} W_{34})\\
    &- (W_{24}\mathcal{I}(\{1,3\}|\varnothing)+W_{14} W_{23}+W_{12} W_{34})\\
    &=V_2 V_4 W_{13}-W_{24} W_{13}\\[5pt]
    \operatorname{KN}\cdot V_2 V_4 W_{13}=\frac{1}{\alpha'}(\partial_{z_2}\operatorname{KN})\cdot (V_4 W_{13})&\xlongequal{\mathrm{IBP}}-\frac{1}{\alpha'}\operatorname{KN}\cdot \partial_{z_2}(V_4 W_{13})=\operatorname{KN}\cdot W_{24}W_{13}\\
        \Rightarrow \mathcal{I}(\{1,3\}|\{2,4\})&\xlongequal{\mathrm{IBP}}\mathcal{I}(\{1,2,3\}|\{4\})
\end{aligned}
\end{equation}

\subsubsection{DR of open superstrings} \label{sec: DR of open superstring}
From \myTheorem~\ref{thm: IBP building block equivalence} we see that for any set $\mathrm{T}$, the integral of $\mathcal{I}(\mathrm{T}\cap\text{I}|\mathrm{T}\cap\text{II})$ is independent of split of $m=n-2$ particles $\{3,\ldots,n\}$ into $\mathrm{I}\sqcup\mathrm{II}$. Thus if the DR of superstring correlator can be reformulated as linear combination of $\mathcal{I}(\mathrm{T}\cap\text{I}|\mathrm{T}\cap\text{II})$s, then the resulting stringy NLSM is independent of how we choose the split in general DR~\eqref{eq: DRgeneral}, and it is the case.

\begin{\mytheorem}
\label{thm: superstring linear expand}
The dimensional reduced open superstring correlator can be written as
\begin{equation}\label{eq: expand DR superstring correlator}
    \varphi^{\mathrm{scalar}}_{n} =\mathcal{DR}\left(\partial_{e_1\cdot e_2}\varphi^{\mathrm{gauge}}_{n}\right) =\frac{2(-1)^{\lfloor\frac{n}{2}\rfloor-1}}{z_{12}^2}\sum_{\substack{\mathrm{T}\subset\{3,\ldots,n\}\\ \left |\mathrm{T}\right |+n=\mathrm{even}}}\mathcal{I}(\mathrm{T}\cap\mathrm{I}\mspace{1mu}|\mspace{1mu}\mathrm{T}\cap\mathrm{II})\,\mathrm{det}\mspace{1mu}(\mathbf{A}_{(12)\mspace{1mu}\sqcup\mspace{1mu}\mathrm{T}}),
\end{equation}
where $\mathbf{A}_\mathrm{T}$ denotes the matrix $\mathbf{A}$ defined in~\eqref{def: CHY building blocks} with its $i$-th columns and rows removed for $\forall\, i\in \mathrm{T}$, hence $\mathrm{det}(\mathbf{A}_{(12)\mspace{1mu}\sqcup\mspace{1mu}\mathrm{T}})$ is independent of $z_i\mspace{1mu}|_{\mspace{2mu}i\mspace{1mu}\in\mspace{1mu}(12)\mspace{1mu}\sqcup\mspace{1mu}\mathrm{T}}$.
\end{\mytheorem}
For brevity, we will only discuss the application of \myTheorem~\ref{thm: superstring linear expand} here, leaving its proof to Appendix~\ref{app.: Supplement for superstring theory}. Let $\text{I}=\varnothing$, the resulting stringy NLSM amplitude is given by
\begin{equation}\label{eq: IBP@DR superstring correlator}
\begin{split}
    \mathcal{M}^{\text{NLSM}}_{n}=&
    \int d\mu_n^\text{string}~\frac{2(-1)^{\lfloor\frac{n}{2}\rfloor-1}}{z_{12}^2}\sum_{\substack{\mathrm{T}\subset\{3,\ldots,n\}\\ \left |\mathrm{T}\right |+n=\mathrm{even}}}\mathcal{I}(\varnothing\mspace{1mu}|\mspace{1mu}\mathrm{T})\,\mathrm{det}\mspace{1mu}(\mathbf{A}_{(12)\mspace{1mu}\sqcup\mspace{1mu}\mathrm{T}})\\
    =&\int d\mu_n^\text{string}~\frac{2}{z_{12}^2}\operatorname{Pf}\begin{pmatrix}\mathbf{A}&\mathbf{E}\\-\mathbf{E}&-\mathbf{A}\end{pmatrix}^{12}
    =\int d\mu_n^\text{string}~\frac{2}{z_{ij}^2}\operatorname{Pf}\begin{pmatrix}\mathbf{A}&\mathbf{E}\\-\mathbf{E}&-\mathbf{A}\end{pmatrix}^{ij},
\end{split}
\end{equation}
where $\mathbf{E}$ is the $n\times n$ diagonal matrix defined as $(\mathbf{E})_{ij}=\delta_{ij}V_i$, the second line comes from the expansion of the Pfaffian in $V_i$~\eqref{eq: Pf expansion in V}. A even briefer expression of stringy NLSM amplitude corresponding to $\text{II}=\varnothing$ is given by
\begin{equation}\label{eq: (super) stringy nlsm}
    \mathcal{M}^{\text{NLSM}}_{n}
    =\int d\mu_n^\text{string}~\frac{2}{z_{ij}^2}\operatorname{det}(\mathbf{\tilde A}_{ij}),\qquad \mathbf{\tilde A}:=\mathbf{A}-\mathbf{E}.
\end{equation}

Here we give an explicit example for $n=5$, $\text{I}=\{3\}$, $\text{II}=\{4,5\}$:
\begin{equation}
\begin{aligned}
    \varphi^{\mathrm{scalar}}_{5}&=\frac{-2}{z_{12}^2}\left(
    \begin{aligned}
        &\det\left(\mathbf{A}_{124}\right) \mathcal{I}(\varnothing|\{4\})+\det\left(\mathbf{A}_{125}\right) \mathcal{I}(\varnothing|\{5\})+\\
        \,+&\det\left(\mathbf{A}_{123}\right) \mathcal{I}(\{3\}|\varnothing)+\det\left(\mathbf{A}_{12345}\right) \mathcal{I}(\{3\}|\{4,5\})\,
    \end{aligned}\right)\\
    &=\frac{-2}{z_{12}^2}\left(
        \det\left(\mathbf{A}_{124}\right)V_4 +\det\left(\mathbf{A}_{125}\right)V_5 +\det\left(\mathbf{A}_{12345}\right) \left(V_5 W_{34}+V_4 W_{35}\right) \right)\\
    &=\frac{2}{z_{12}^2}\times\underset{\sharp}{\underbrace{-\left((\mathbf{A})_{35}^2 V_4+(\mathbf{A})_{34}^2 V_5+V_5 W_{34}+V_4 W_{35}\right)}}\,,\\[3pt]
    \sharp&=-\left((\mathbf{A})_{35}^2 V_4+(\mathbf{A})_{34}^2 V_5-(\alpha')^{-1}\partial_{z_3}(V_4 V_5)\right)\\
    &\xlongequal{\text{IBP}}-\left((\mathbf{A})_{35}^2 V_4+(\mathbf{A})_{34}^2 V_5+V_3 V_4 V_5\right)
    =\operatorname{Pf}\begin{pmatrix}\mathbf{A}&\mathbf{E}\\-\mathbf{E}&-\mathbf{A}\end{pmatrix}^{12}.
\end{aligned}
\end{equation}

This odd-$n$ integral, as expected, finally evaluates to zero.

\section{Mixed amplitudes and logarithmic forms}\label{sec: Mixed amplitudes and logarithmic forms}
In this section, we extend the stringy NLSM~\eqref{eq: (super) stringy nlsm} to give the stringy UV completion of the mixed amplitudes of $\phi^3$ and pions scattering. We also develop a systematic method to compute its low-energy limit by giving the logarithmic form of the string integral, thus the corresponding $\alpha^\prime$ order can be obtained by plugging in the result of Z-integral~\cite{Mafra:2016mcc}.

\addtolength{\abovedisplayskip}{-0.15\baselineskip}

To begin with, we shall give the superstring YMS integrand of the mixed amplitudes of gluons and bi-adjoint scalars. As suggested in~\cite{Cheung:2017ems}, for the scattering of gluons $\{i_1,i_2,\ldots,i_m\}$ and scalars ordered by $(1,\alpha,n)$, it is natural to identify the superstring integrand as
\begin{equation}\label{eq: gauge+color definition}
    \varphi^{\text{gauge}+\text{color}}_{n}(\{i_1,i_2,\ldots,i_m\}|1,\alpha,n)=\partial_{\mspace{2mu}2 e_{1} \cdot e_{n}} \prod_{i=1}^{|\alpha|} \partial_{\mspace{2mu}2 e_{\alpha_i} \cdot (p_{\alpha_{i-1}} - p_{n}) } \varphi^{\text{gauge}}_{n},
\end{equation}
where we define $\alpha_0:=1$ and choose $(i_0,j_0)=(1,n)$ in~\eqref{eq: superstring correlator} to match the conventions in~\cite{Cheung:2017ems}. Algebraic operations (see Appendix~\ref{app.: Supplement for mixed amplitude} for details) show that 
\begin{equation}\label{eq: gauge+color correlator}
\begin{split}
\hspace{-0.5em}\varphi^{\text{gauge}+\text{color}}_{n}(\{i_1,\ldots,i_m\}|1,\alpha,n) &=\mathrm{PT}(1,\alpha,n)\! \sum_{q=0}^{\lfloor m/2 \rfloor}\! (-\alpha')^{-q}\! \sum_{\rho\in\mathfrak{S}_q} \prod_{(ij)\in\rho} \!\frac{2 e_{i}\mspace{-1mu}\cdot\mspace{-1mu} e_{j}}{z_{i j}^2}\,\operatorname{Pf}\mathbf{\Psi}^{1,\alpha,\rho,n}\mspace{1mu},\\[2pt]
\mathrm{PT}(\alpha):\mspace{-5mu}&=\frac{1}{z_{\alpha(1)\alpha(2)} z_{\alpha(2)\alpha(3)}\cdots z_{\alpha(n)\alpha(1)}}\,,
\end{split}
\raisetag{35pt}
\end{equation}
where $\mathfrak{S}_q$, as defined in Section~\ref{sec: Adler zero from gauge invariance}, only contains gauge particles, while $\mathbf{\Psi}^{1,\alpha,\rho,n}$ denotes the matrix $\mathbf\Psi$ with its $i,n{+}i$-th columns and rows removed for each $i$ in $(1,\alpha,n)$ or $\rho$.

\addtolength{\abovedisplayskip}{0.15\baselineskip}

Now \myTheorem~\ref{thm: equivalent DR} guarantees that for any split of the $m$ gauge particles, dimensional reduction~\eqref{eq: DRgeneral} gives the same result. This result is then naturally identified as the stringy correlator of the scattering of pions $\{i_1,i_2,\ldots,i_m\}$ and colored scalars $(1,\alpha,n)$:
\begin{equation}
    \varphi^{\text{gauge}+\text{color}}_{n}(\{i_1,i_2,\ldots,i_m\}|1,\alpha,n)\;  \xrightarrow{\,\eqref{eq: DRgeneral}\,}\; \varphi^{\text{scalar}+\text{color}}_{n}(\{i_1,i_2,\ldots,i_m\}|1,\alpha,n).
\end{equation}

Let $\text{I}=\varnothing$ in DR~\eqref{eq: DRgeneral}, the stringy integrand for mixed amplitude reduces to:
\begin{equation}\label{eq: mixed amplitude in det(A) PT}
    \varphi^{\text{scalar}+\text{color}}_{n}(\{i_1,i_2,\ldots,i_m\}|1,\alpha,n)=\operatorname{det}\mathbf{\tilde A}_{\{i_1,i_2,\ldots,i_m \}} \operatorname{PT}(1,\alpha,n),
\end{equation}
where we use $\mathbf{\tilde A}_{\{i_1,i_2,\ldots,i_m \}}$ with braces to denote $\mathbf{\tilde A}$ defined in~\eqref{eq: (super) stringy nlsm} with only columns and rows in $\{i_1,i_2,\ldots,i_m \}$ \textit{left}. Note that for $\alpha=\varnothing$, it reduces to the pure pions scattering.

\addtolength{\abovedisplayskip}{-0.15\baselineskip}

Our next task is to expand~\eqref{eq: mixed amplitude in det(A) PT} into a more explicit form. By fixing $z_n\to\infty$ and omitting all $z_{i n}\to\infty$, we can draw an analogy to the matrix tree theorem~\cite{Feng:2012sy} to derive
\begin{equation}\label{eq: pion+color correlator}
    \operatorname{det}\mathbf{\tilde A}_{\{i_1,i_2,\ldots,i_m \}}\operatorname{PT}(1,\alpha,n) = (-1)^{n}\sum_{G(1,\alpha)} \prod_{e(i,j)} \frac{s_{ij}}{z_{ij}} \prod_{k=1}^{|\alpha|} \frac{1}{s_{\alpha_{k-1}\alpha_{k}} },
\end{equation}
where the summation goes over all labelled trees $G(1,\alpha)$ containing the sub-tree $(1,\alpha)$, with nodes $\{1,2,\ldots,n-1\}$ and orientations of the $n-2$ edges $e(i,j)$ flowing to the root node $1$. For the definition of labelled tree, see~\cite{Gao:2017dek}. Proof of~\eqref{eq: pion+color correlator} will be given in Appendix~\ref{app.: Supplement for mixed amplitude}.

\addtolength{\abovedisplayskip}{0.15\baselineskip}

For instance, let $n=5$ with particles $3,4$ to be pions and omit $z_{i n}\to\infty$, we have:
\begin{equation} \label{eq: MTexample}
\begin{aligned}
\operatorname{det}\mathbf{\tilde A}_{\{3,4 \}} \operatorname{PT}(1,2,5)
&= \frac{s_{13} s_{14}}{z_{12} z_{31} z_{41}}+\frac{s_{23} s_{14}}{z_{12} z_{32} z_{41}}+\frac{s_{34} s_{14}}{z_{12} z_{34} z_{41}}+\frac{s_{13} s_{24}}{z_{12} z_{31} z_{42}}\\
&+\frac{s_{23} s_{24}}{z_{12} z_{32} z_{42}}+\frac{s_{24} s_{34}}{z_{12} z_{34} z_{42}}+\frac{s_{13} s_{34}}{z_{12} z_{31} z_{43}}+\frac{s_{23} s_{34}}{z_{12} z_{32} z_{43}}.
\end{aligned}
\end{equation}

Now we are just one step away from the logarithmic forms we expect. Notice that the denominators in~\eqref{eq: pion+color correlator} are Cayley functions (generalization of PT factors) defined in~\cite{Gao:2017dek}, thus we can apply (3.3) in~\cite{Gao:2017dek} to expand these denominators to Kleiss-Kuijf (KK) basis. This yields the logarithmic forms for mixed scattering of pions and colored scalars:
\begin{equation} \label{eq: BCJnumerator}
    \varphi^{\text{scalar}+\text{color}}_{n}(\{i_1,i_2,\ldots,i_m\}|1,\alpha,n)
    =\sum_{\substack{\sigma\in \mathrm{Perm}(i_1,\ldots,i_m)\\[1pt] \rho=\alpha \shuffle \sigma}} \prod_{a=1}^{m} s_{i_a|1\rho_{:i_a}} \operatorname{PT}(1,\rho,n),
\end{equation}
where the first summation goes over the permutations of the pion list $\sigma$, and the second summation goes over the shuffle product of $\alpha$ and $\sigma$. Here $s_{i_a|J}:=2\mspace{2mu}p_{i_a} \cdot (\sum_{j\in J}p_j)$, while $1\rho_{:i_a}$ denotes the elements in $\{1\}\sqcup\rho$ from the beginning to $i_a$, \textit{e.g.} $(1,2,5,4,3)_{:5}=(1,2,5)$. We remark that for $\alpha=\varnothing$, this is exactly the well-known BCJ numerator for pure pions scattering~\cite{Du:2016tbc,Carrasco:2016ldy}, see also~\cite{He:2021lro,Cheung:2021zvb}. 
The above example~\eqref{eq: MTexample} can be now written as
\begin{equation}
\begin{aligned}
    \varphi^{\text{scalar}+\text{color}}_{n}(\{3,4\}|1,2,5)
    &=s_{3|12}\, s_{4|123}\operatorname{PT}(1,2,3,4,5)+ s_{13}\, s_{4|123}\operatorname{PT}(1,3,2,4,5) \\
    &+ s_{13}\, s_{4|13}\operatorname{PT}(1,3,4,2,5)+ s_{4|12}\, s_{3|124}\operatorname{PT}(1,2,4,3,5)\\
    &+ s_{14}\, s_{3|124}\operatorname{PT}(1,4,2,3,5)+ s_{14}\, s_{3|14}\operatorname{PT}(1,4,3,2,5).
\end{aligned}
\end{equation}

The closed formula of the logarithmic form~\eqref{eq: BCJnumerator} provides a powerful algorithm to compute the amplitudes in $\alpha^\prime \to 0$ limit: we just need to plug in the corresponding $\alpha^\prime$ order of Z-integral, which have been computed in~\cite{Mafra:2016mcc}. In particular, for the leading field theory limit, we plug in the well-known bi-adjoint $\phi^3$ amplitudes $m(1,2,\ldots,n|1,\rho,n)$~\cite{Cachazo:2013iea} for each PT$(1,\rho,n)$. For example, the field theory limit of~\eqref{eq: MTexample} simply reads
\begin{equation} \label{eq:mixedEx1}
\begin{aligned}
    A^{\mathrm{NLSM}+ \phi^3}_5(\{3,4\}|1,2,5)=-1+\frac{X_{1,4}}{X_{1,3}}+\frac{X_{3,5}}{X_{1,3}}+\frac{X_{3,5}}{X_{2,5}}+\frac{X_{2,4}}{X_{2,5}}.
\end{aligned}
\end{equation}
where we use the planar variables $X_{i,j}:= (p_i+p_{i+1}+\cdots+p_{j-1})^2$ defined in~\cite{Arkani-Hamed:2017mur}. Another example for $n=8$ in the field theory limit with $4$ disjointed pions is given by
\begin{equation*}
\begin{aligned}
    A^{\mathrm{NLSM}+ \phi^3}_8(\{1,3,5,7\}|2,4,6,8)&=\frac{1}{2 X_{1,5}}+\frac{1}{2 X_{2,6}}-\frac{X_{1,3}+X_{2,4}}{X_{1,4} X_{1,5}}-\frac{X_{1,7}+X_{2,8}}{X_{2,6} X_{2,7}}\\[2pt]
    &\hspace{-3em}+\frac{\left(X_{1,7}+X_{2,8}\right) \left(X_{3,5}+X_{4,6}\right)}{2 X_{2,6} X_{2,7} X_{3,6}}+\frac{\left(X_{1,7}+X_{2,8}\right) \left(X_{3,5}+X_{4,6}\right)}{2 X_{2,7} X_{3,6} X_{3,7}}\\[2pt]
    &\hspace{-3em}+ (\text{cyclic}\; i \to i+2,i+4,i+6),
\end{aligned}
\end{equation*}
where we need to plug in a slightly different version of~\eqref{eq: BCJnumerator} with {\it e.g.} $(i_0,j_0)=(2,8)$ to be special instead of $(1,8)$. We also present an example for 9-point with 6 pions in Appendix~\ref{app.: Supplement for mixed amplitude}.

Apart from the field theory limit, the higher $\alpha'$ order can also be worked out through this method, \textit{e.g.} the $O(\alpha'\mspace{1mu}^2)$ correction of~\eqref{eq:mixedEx1} is given by: 
\begin{equation*}
\begin{aligned}
    \frac{\pi^2}{6}  \left(-\frac{X_{3,5} X_{1,4}^2}{X_{1,3}}-X_{2,5} X_{1,4}+X_{3,5} X_{1,4}-\frac{X_{3,5}^2 X_{1,4}}{X_{1,3}}+X_{3,5}^2-X_{1,3} X_{2,4}\right. \ \ &\\
    \left. +X_{1,3} X_{2,5}-X_{1,3} X_{3,5}+X_{2,4} X_{3,5}-X_{2,5} X_{3,5}-\frac{X_{2,4} X_{3,5}^2}{X_{2,5}}-\frac{X_{2,4}^2 X_{3,5}}{X_{2,5}}\right)&.
\end{aligned}
\end{equation*}

\section{Pions in open bosonic strings and closed strings}\label{sec: Pions in open bosonic strings and closed strings}
In this section, we first demonstrate that the stringy NLSM corresponding to the bosonic string is regardless of the different DR but the first derivative $\partial_{e_i\cdot e_j}$ in the same IBP-based method. Then, we generalize our discussion to closed string models for both the bosonic and super string. We also discuss other gauge theories/EFTs obtained via the KLT relation. We list some of their 4-point results as a simple comparison between the two versions of stringy NLSM.

\addtolength{\abovedisplayskip}{-0.2\baselineskip}
\addtolength{\belowdisplayskip}{-0.2\baselineskip}

\subsection{DR of bosonic strings}
The bosonic string correlator $\varphi_{\pm,n}^{\text{bosonic}}$~\cite{He:2019drm,Schlotterer:2016cxa} can be written as a formula very similar to the IBP building blocks~\eqref{def: IBP building blocks} in our $\mathfrak{S}$-notation:
\begin{equation}\label{def: bosonic string correlator}
    \varphi_{\pm,n}^{\text{bosonic}}=\sum_{q=0}^{\lfloor n/2 \rfloor} \sum_{\rho\in\mathfrak{S}_q}\prod_{(ij)\in \rho}\mathcal{B}^{\,\pm}_{(ij)} \; \prod_{k\in\bar{\rho}}\mathcal{B}_{(k)}.
\end{equation}

\addtolength{\abovedisplayskip}{0.1\baselineskip}
\addtolength{\belowdisplayskip}{0.1\baselineskip}

Here the two types of letters $\mathcal{B}_{(i)},\mathcal{B}^{\,\pm}_{(ij)}$ are defined as:
\begin{equation}
	\mathcal{B}_{(i)}:=\sum_{k\neq i} \frac{2e_i \cdot p_j}{z_{ij}},\qquad \mathcal{B}^{\,\pm}_{(ij)}:=\pm\frac{2e_i \cdot e_j}{\alpha'z_{ij}^2}.
\end{equation}

\addtolength{\abovedisplayskip}{0.1\baselineskip}
\addtolength{\belowdisplayskip}{0.1\baselineskip}

This subsection only involves $\varphi_{+,n}^{\text{bosonic}}$, which is abbreviated to $\varphi_n^{\text{bosonic}}$. Similarly we abbreviate $\mathcal{B}^{\,+}_{(ij)}$ to $\mathcal{B}_{(ij)}$. For example, the bosonic string correlator for $n=4$ is given by:
\begin{equation*}
\begin{aligned}
	\varphi_4^{\text{bosonic}}&=\mathcal{B}_{(12)} \mathcal{B}_{(3)} \mathcal{B}_{(4)} +\mathcal{B}_{(13)} \mathcal{B}_{(2)} \mathcal{B}_{(4)} +\mathcal{B}_{(14)} \mathcal{B}_{(2)} \mathcal{B}_{(3)} \\
    &+\mathcal{B}_{(23)} \mathcal{B}_{(1)} \mathcal{B}_{(4)} +\mathcal{B}_{(24)}\mathcal{B}_{(1)} \mathcal{B}_{(3)} +\mathcal{B}_{(34)}\mathcal{B}_{(1)} \mathcal{B}_{(2)}\\
    &+\mathcal{B}_{(12)} \mathcal{B}_{(34)} +\mathcal{B}_{(13)} \mathcal{B}_{(24)} +\mathcal{B}_{(14)} \mathcal{B}_{(23)} +\mathcal{B}_{(1)} \mathcal{B}_{(2)} \mathcal{B}_{(3)} \mathcal{B}_{(4)}
\end{aligned}
\end{equation*}

To prove the equivalence of DRs of bosonic string correlator with respect to the different split of gauge particles into $\text{I}\sqcup\text{II}$, we should express DRs of $\varphi_n^{\text{bosonic}}$ as linear combinations of IBP building blocks, then apply \myTheorem~\ref{thm: IBP building block equivalence}. We find it much simpler than the superstring cases, since the DRs of $\mathcal{B}_{(i)}$ and $\mathcal{B}_{(ij)}$ are exactly $V_i$ and $W_{i,j}$ defined in~\eqref{def: VW}:
\begin{equation}
	\mathcal{DR}(\mathcal{B}_{(i)})=
	\begin{dcases}
		0 & i\in\text{I}\\
		V_i & i\in\text{II}
	\end{dcases} ,\qquad
	\mathcal{DR}(\mathcal{B}_{(ij)})=
	\begin{dcases}
		W_{ij} & \{i,j\}\nsubseteq\text{II}\\
		0 & \{i,j\}\subseteq\text{II}
	\end{dcases}
\end{equation}

Comparing with the definition of IBP building blocks~\eqref{def: IBP building blocks}, one can easily see that:
\begin{equation}\label{eq: boscalar}
	\mathcal{DR}\left(\partial_{e_i\cdot e_j}\varphi_n^{\text{bosonic}}\right)=\frac{2}{\alpha' z_{ij}^2}\mathcal{I}(\text{I}|\text{II}),\qquad \text{I}\sqcup\text{II}=\{1,2,\dots,n\}\backslash\{i,j\}.
\end{equation}

For example, for the 6-point amplitude with first derivative $\partial_{e_5\cdot e_6}$ we have:
\begin{equation}
\begin{aligned}
	&\mathcal{DR}^{\text{II}=\{2,4\}}(\partial_{e_5\cdot e_6}\varphi_6^{\text{bosonic}})
    =\mathcal{DR}^{\text{II}=\{2,4\}}(\frac{2}{\alpha' z_{56}^2}\varphi_4^{\text{bosonic}})\\
    &=\frac{2}{\alpha' z_{56}^2}\left(V_2 V_4 W_{13}+W_{14} W_{23}+W_{12} W_{34}\right)=\frac{2}{\alpha' z_{56}^2}\mathcal{I}(\{1,3\}|\{2,4\}).
\end{aligned}
\end{equation}

Applying \myTheorem~\ref{thm: IBP building block equivalence}, we reach our desired conclusion that the general DRs~\eqref{eq: DRgeneral} of the bosonic string given the first derivative are equivalent, regardless of how we split $m$ gauge particles into $\text{I}\sqcup\text{II}$. However, the DRs with different first derivative $\partial_{e_i \cdot e_j}$ can be different, since there is no equivalence between $\mathcal{I}(\text{I}|\text{II})$ and $\mathcal{I}(\text{I}'|\text{II}')$ for $\text{I}\sqcup\text{II}\neq\text{I}'\sqcup\text{II}'$. For example, the integrated results of taking $\partial_{e_1 \cdot e_2}$ or $\partial_{e_1 \cdot e_4}$ at $n=4$ are:
\begin{equation}
\begin{aligned}
   e_1\cdot e_2:~&-\frac{s_{1 2} \Gamma \left(s_{1 2}-1\right) \Gamma \left(s_{2 3}+1\right)}{\Gamma \left(s_{1 2}+s_{2 3}\right)},\\
   e_1\cdot e_4:~&-\frac{s_{2 3} \Gamma \left(s_{1 2}+1\right) \Gamma \left(s_{2 3}-1\right)}{\Gamma \left(s_{1 2}+s_{2 3}\right)},
\end{aligned}
\end{equation}
where we set $\alpha'=1$. We can interpret them respectively as the scattering of $(1^\phi,2^\phi,3^\pi,4^\pi)$ and $(1^\phi,2^\pi,3^\pi,4^\phi)$. In the field theory limit, both of them are known to be equivalent to the pure pions scattering, which reads $s_{1 2}+s_{2 3}$. However, their higher $\alpha'$ order corrections are no longer equivalent. We remark that the dependency on the first derivative is the characteristic that distinguishes DRs of bosonic string from that of superstring.

To end this subsection, let us give a comparison between \myTheorem~\ref{thm: equivalent DR} and the IBP proof based on \myTheorem~\ref{thm: IBP building block equivalence}. For an integrated gauge invariant amplitude, the difference between DRs is some vanishing gauge terms, \textit{i.e.} Schwinger terms that does not contribute to the amplitude. While for the stringy disk integral, the difference between DRs are some vanishing boundary terms due to the short-distance behavior of the Koba-Nielsen factor.

\subsection{Closed super and bosonic strings}
There is a natural generalization of what we have considered, {\it i.e.} the closed super and bosonic string models. As studied in~\cite{Mizera:2019blq,Mizera:2019gea} and we have discussed, we have the following left (right) movers
\begin{gather}
    \varphi^\mathrm{gauge}_{\pm}=\eqref{eq: superstring correlator} \qquad \xrightarrow[\text{take}~e_i\cdot e_j]{\text{DR}} \qquad \varphi^\mathrm{scalar}_{\pm}=\mathrm{det}'\mathbf{\tilde A}\mspace{2mu},\\[5pt]
     \varphi^\mathrm{bosonic}_{\pm}=\eqref{def: bosonic string correlator}\qquad \xrightarrow[\text{take}~e_i\cdot e_j]{\text{DR}} \qquad \varphi^\mathrm{boscalsr}_{\pm}=\eqref{eq: boscalar}\mspace{2mu},\\[5pt]
    \varphi^\mathrm{color}_{\pm}= \operatorname{PT}(1,\dots,n)\mspace{2mu}.\phantom{\frac{1}{1}} 
\end{gather}

The closed string amplitudes are then given by the modulus squared integral
\begin{equation}
    \mathcal{M}_n^{\text {closed}}=\int_{\mathbb{C}^n}\frac{d^n z d^n \bar{z}}{\operatorname{vol~SL}(2, \mathbb{C})}\prod_{i<j}(\left|z_{i j}\right|^2)^{\alpha's_{i j}}  \varphi_{-,n}(z)\varphi_{+,n}(\bar{z})
\end{equation}

For different combinations of the left (right) movers, the corresponding amplitudes are listed in Table~\ref{tab: different combination of movers}, in which it is straightforward to see that the Born-Infeld (special Galileon) is given by dimensional reductions of Einstein gravity on its left (left and right) movers.

\begin{table}[htbp] 
\centering
\resizebox{\columnwidth}{!}{
\subfloat[theories in "gauge family"]{ \begin{tabular}{l|ccccc}
    & $\varphi_{+}^{\mathrm{gauge}}$ & $\varphi_{+}^{\mathrm{scalar}}$ \\ [5pt] \hline \\ [0.5pt]
    $\varphi_{-}^{\mathrm{gauge}}$ & Einstein gravity &  \\ [5pt]
    $\varphi_{-}^{\mathrm{scalar}}$ & Born-Infeld & special Galileon~\cite{Cachazo:2014xea} \\ [5pt]
    $\varphi_{-}^{\mathrm{color}}$ & Yang-Mills & NLSM \\ [5pt]
\end{tabular}}
\qquad
\subfloat[theories in "bosonic family"]{ \begin{tabular}{l|ccccc}
& $\varphi_{+}^{\mathrm{bosonic}}$ & $\varphi_{+}^{\mathrm{boscalar}}$  \\ [5pt] \hline \\ [0.5pt]
    $\varphi_{-}^{\mathrm{gauge}}$ & Weyl-Einstein gravity~\cite{Johansson:2017srf, Azevedo_2018}&  \\ [5pt]
    $\varphi_{-}^{\mathrm{boscalar}}$ & bos-BI & bos-sGal \\ [5pt]
    $\varphi_{-}^{\mathrm{color}}$ & YM+$(DF)^2$~\cite{Johansson:2017srf, Azevedo_2018}& bos-NLSM  \\ [5pt]
\end{tabular}}
}
\caption{Various theories from different combinations of left (right) movers, where bos-BI/sGal/ NLSM are the models containing building blocks from (DRs of) bosonic string ones.}
\label{tab: different combination of movers}
\end{table}

Noteworthy, the open and closed bi-adjoint scalar amplitudes are related by the single value projections~\cite{Stieberger:2013wea,Stieberger:2014hba,Brown:2018omk,Schlotterer:2018zce}. As a consequence, the ordered amplitudes of open and closed string are related via the single value projection since one can always express the $\varphi_{+}$ as logarithmic forms on the support of Integration-By-Part relations~\cite{He:2018pol,He:2019drm}, and can therefore expand it into the Parke-Taylor factor $\varphi^\mathrm{color}_{\pm}$.

\begin{table}[htbp]\centering
\begin{tabular}{|c|c|c|}
    \hline 
    \rule{0pt}{11pt} NLSM & open string & closed string \\ \hline
    \rule{0pt}{14pt} superstring & $\frac{-2\Gamma(1+s)\Gamma(1+t)}{\Gamma(-u)}$ & $\frac{-2\pi\Gamma(1+s)\Gamma(1+t)\Gamma(1+u)}{\Gamma(1-s)\Gamma(1-t)\Gamma(-u)}$ \\[4pt] \hline
   \rule{0pt}{14pt} bosonic string & $\frac{8 \Gamma \left(1+s\right) \Gamma \left(1+t\right)}{(1-s) \Gamma \left(-u\right)}$ & $\frac{8 \pi  \Gamma \left(1+s\right) \Gamma \left(1+t\right) \Gamma (1+u)}{(1-s) \Gamma (1-s) \Gamma (1-t) \Gamma \left(-u\right)}$ \\[4pt] \hline
\end{tabular}
\caption{The 4-point NLSM amplitudes in the open/closed bosonic and super string models, where we choose the first derivative $\partial_{e_1\cdot e_2}$ before the dimensional reductions.}
\label{tab: double-copy}
\end{table}

We present the explicit results of the 4-point stringy NLSM amplitudes in Table~\ref{tab: double-copy}. It is straightforward to show 
the corresponding open and closed string amplitudes are related by the Single-Valued (SV) map~\cite{Stieberger:2013wea,Stieberger:2014hba,Brown:2018omk,Schlotterer:2018zce}: we can write the stringy NLSM amplitudes from the DRs of the open and closed superstring ones as:
\begin{align}
    \mathcal{M}^{\text{NLSM}}_{\text{open}}(1,2,3,4)&=2u e^{o(s)+o(t)+o(u)-e(s)-e(t)+e(u)},\\[5pt]
    \mathcal{M}^{\text{NLSM}}_{\text{closed}}(1,2,3,4)&=2\pi u e^{2o(s)+2o(t)+2o(u)},
\end{align}
where $o(z)$ and $e(z)$ are summations of infinite series of $\zeta(\text{odd})$ and $\zeta(\text{even})$ respectively:
\begin{equation}
	o(z)=\sum_{k=1}^{\infty}\frac{\zeta(2k+1)z^{2k+1}}{2k+1}\qquad e(z)=\sum_{k=1}^{\infty}\frac{\zeta(2k)z^{2k}}{2k}.
\end{equation}
Note that the SV map of the multiple zeta values used here is simply $\zeta(2k+1)\to 2\zeta(2k+1),$ $\zeta(2k)\to 0$, therefore we have:
\begin{equation}
    \pi\left(\mathcal{M}^{\text{NLSM}}_{\text{open}}\right)_{\text{sv}} =\pi \mathcal{M}^{\text{NLSM}}_{\text{open}}\vert_{o(z)\to2o(z),e(z)\to0}=\mathcal{M}^{\text{NLSM}}_{\text{closed}}.
\end{equation}

Similar relations also hold for the NLSM amplitudes as the DRs of the bosonic open and closed string ones. Let us end this section by presenting the results of 4-point sGal and BI in Table~\ref{tab: KLT}, which is computed via the KLT double copy relations. 
\begin{table}[htbp]\centering
\begin{tabular}{|c|c|c|c|c|}
    \hline \rule{0pt}{11pt} Theory & sGal$\sim$NLSM$\otimes$NLSM & BI$\sim$NLSM$\otimes$YM \\ \hline
    \rule{0pt}{14pt} superstring & $\frac{4 \pi  \Gamma (1+s) \Gamma (1+t) \Gamma (1+u)}{\Gamma (-s) \Gamma (-t) \Gamma (-u)}$ & $\frac{16 \pi  \Gamma (1+s) \Gamma (1+t) \Gamma (1+u)}{\Gamma (1-s) \Gamma (1-t) \Gamma (1-u)} A^\text{BI}$ \\[4pt] \hline
    \rule{0pt}{14pt} bosonic string & $\frac{64 \pi  \Gamma (1+s) \Gamma (1+t) \Gamma (1+u)}{(1-s)^2\Gamma (-s) \Gamma (-t) \Gamma (-u)}$ & a complicated expression \\[4pt] \hline
\end{tabular}
\caption{The 4-point amplitudes of sGal and BI in the closed bosonic/super string models, where we choose the first derivative $\partial_{e_1\cdot e_2}$ before DRs and $A^\text{BI}$ is the 4-point field theory amplitude of BI.}
\label{tab: KLT}
\end{table}

\section{Conclusions and outlook}
In this paper we revisit the slogan ``pions as dimensional reduction of gluons''. We generalize the DRs found in the literature and prove that all dimensional reductions we introduce yield the same result for any gauge invariant object multi-linear in polarization vectors. Furthermore, we prove that the resulting function must have Adler's zero if the gauge invariant object only contains local simple poles. We present some explicit examples of open and closed super and bosonic string models, illustrating the equivalence among the DRs via the IBP relations. Our claims can also be applied to mixed amplitudes, for which we derive a closed formula for logarithmic correlators of any number of pions and $\phi^3$ scattering in the superstring induced model, providing a systematic way to compute such amplitudes at the $\alpha'$ expansion.


Our work has suggested several future directions. As we have seen, gauge invariance induces the Adler zero via the dimensional reduction. This connection is quite remarkable since the gauge invariance and Adler zero can be seen as the defining property for gauge theories and EFTs~\cite{Arkani-Hamed:2016rak,Rodina:2016jyz}. However, apart from our primary example, {\it i.e.} the Yang-Mills theory, there are other gauge invariant starting points such as Einstein-Maxwell-scalar and gravity, which reduce into Dirac-Born-Infeld and special Galileon that have enhanced Adler zero~\cite{Cheung:2014dqa,Cheung:2018oki}. The enhancement of the soft behaviour of these theories is still lack of a general understanding as the consequence of DRs. It would be desirable to find interpretations from a similar perspective. Moreover, it is well-known that the (enhanced) Adler zero is related to the shift symmetry, further understanding of the general DRs at the Lagrangian level as has been studied in~\cite{Cheung:2017yef} for a special choice of DR would be valuable.

Another direction to explore is to apply the general DRs we found here on the expansion of the YM amplitudes into mixed ampliudes of gluons and $\phi^3$~\cite{Lam:2016tlk,Fu:2017uzt,Du:2017kpo}, which would lead to various expansions of the NLSM as what has been done in~\cite{Dong:2021qai} using the special DR~\cite{Cheung:2017ems,Cheung:2017yef}.

On the other hand, it would be interesting to investigate the relations among the stringy NLSM models~\cite{Mizera:2019blq,Mizera:2019gea} we explored and those proposed in~\cite{Bianchi:2020cfc,Arkani-Hamed:2023swr,Arkani-Hamed:2024nhp,Carrasco:2016ldy,Carrasco:2016ygv}. There are significant and noteworthy differences between these models, for example, the models we study satisfy the monodromy relations~\cite{kawai1986relation,Stieberger:2009hq,Bjerrum-Bohr:2009ulz} while those in~\cite{Carrasco:2016ldy,Carrasco:2016ygv} satisfy the BCJ relations~\cite{Bern:2008qj,Bern:2010ue}.

Moreover, recent studies reveal a series of fantastic behaviors of the string/particle amplitudes, namely smooth splittings~\cite{Cachazo:2021wsz}, zeros and factorizations near the zeros~\cite{Arkani-Hamed:2023swr} (see also~\cite{Bartsch:2024amu, Li:2024qfp}). Later in~\cite{Cao:2024gln} it has been realized that another interesting phenomenon, referred to as the 2-splittings, provides a common origin for both the smooth splittings and the factorizations near zeros. It is natural to wonder whether the stringy models we study share these phenomena. 

Finally, we would like to generalize our understanding to loop amplitudes. At the 1-loop level, it is found in~\cite{Edison:2023ulf,Dong:2023stt} that the special DR~\cite{Cheung:2017ems,Cheung:2017yef} (but without taking out a special pair $e_i \cdot e_j$) acting on the scalar-loop YM integrand, produces the correct NLSM integrand at one loop. Given the fact the results in~\cite{Edison:2023ulf,Dong:2023stt} are based on the forward limit of tree amplitudes, it is evidenced that our general DRs~\eqref{eq: DRgeneral} would produce the same answer. It would be highly desirable to understand dimensional reductions in higher loop level, perhaps one of the available tools is the surfaceology that has been recently studied in~\cite{Arkani-Hamed:2023lbd,Arkani-Hamed:2023mvg,Arkani-Hamed:2023swr,Arkani-Hamed:2023jry,Arkani-Hamed:2024nhp,Arkani-Hamed:2024vna,Arkani-Hamed:2024yvu}.

\acknowledgments
We are very grateful to Song He for proposing this problem. It is our pleasure to thank Qu Cao and Song He for collaborations on the early stage of this project and stimulating discussions. We also thank Canxin Shi, Yao-Qi Zhang and Yong Zhang for discussions and collaborations on related projects.
This work is supported by the National Natural Science Foundation of China under Grant No. 12225510, 11935013, 12047503.

\newpage
\appendix


\section{Details for DRs of superstring theory}\label{app.: Supplement for superstring theory}
In this appendix we will give the proof of \myTheorem~\ref{thm: superstring linear expand}. We will set $(i_0,j_0)=(1,2)$ in~\eqref{eq: superstring correlator}, consistent with Section~\ref{sec: Pions in open superstrings}. Before 
performing dimensional reduction, let us introduce a key tool for handling the reduced Pfaffians in~\eqref{eq: superstring correlator}, namely the following Laplace expansion of the Pfaffian of a skew-symmetric $2n\times2n$ matrix $\bm{\mathcal{A}}\,$:
\begin{equation}\label{eq: expansion of Pf}
    \operatorname{Pf}(\bm{\mathcal{A}}\mspace{1mu}) =\sum_{j=1\mspace{1mu},\mspace{2mu}j\neq i}^{2n}(-1)^{i+j+1+\theta(i-j)}(\bm{\mathcal{A}}\mspace{1mu})_{ij}\operatorname{Pf}(\bm{\mathcal{A}}\mspace{1mu}_{ij})
\end{equation}
where $\theta(i-j)$ is the Heaviside step function, $(\bm{\mathcal{A}}\mspace{1mu})_{ij}$ denotes the matrix element in the $i$-th row and $j$-th column, and $\bm{\mathcal{A}}\mspace{1mu}_{ij}$ denotes $\bm{\mathcal{A}}$ with its $i$-th and $j$-th columns and rows removed, the index $i$ can be chosen arbitrarily.

Now we can apply~\eqref{eq: DRgeneral} to get the dimensional reduced Pfaffians and their prefactors. By further performing a series of elementary row and column transformations on $\mathbf{\Psi}$ without changing its Pfaffian, we obtain a new skew-matrix $\mathbf{\Phi}$ equivalent to $\mathbf{\Psi}$:
\begin{equation}
    \mathbf{\Psi} \xrightarrow{\,\eqref{eq: DRgeneral}\,} \mathbf{\Phi}=\begin{pmatrix}
        \mathbf{A}&\mathbf{D}\\-\mathbf{D}&-\mathbf{A}
    \end{pmatrix},\qquad
    (\mathbf{D})_{ij}:=\begin{dcases}
    V_i&\text{if}\quad i = j\in\text{II}\\ 
    0&\text{otherwise}
    \end{dcases},
\end{equation}

The explicit progress transforming $\mathbf{\Psi}$ to $\mathbf{\Phi}$ is given by:
\begin{equation}
\begin{aligned}
    \mathbf{\Psi}^{12}\xrightarrow{\,\eqref{eq: DRgeneral}\,} &\begin{pmatrix}
        \mathbf{A}_{\{\text{I},\text{I}\}}&\mathbf{A}_{\{\text{I},\text{II}\}}  &\mathbf{0}  &\mathbf{A}_{\{\text{I},\text{II}\}} \\ 
        \mathbf{A}_{\{\text{II},\text{I}\}} &\mathbf{A}_{\{\text{II},\text{II}\}}  &\mathbf{0}  &(\mathbf{A}+\mathbf{D})_{\{\text{II},\text{II}\}}\\
        \mathbf{0} &\mathbf{0}  &-\mathbf{A}_{\{\text{I},\text{I}\}}  &-\mathbf{A}_{\{\text{I},\text{II}\}} \\ 
        \mathbf{A}_{\{\text{II},\text{I}\}} &(\mathbf{A}-\mathbf{D})_{\{\text{II},\text{II}\}}  &-\mathbf{A}_{\{\text{II},\text{I}\}}  &\mathbf{0} 
    \end{pmatrix}\\[5pt]
    =&\begin{pmatrix}
        \mathbf{A}_{\{\text{I},\text{I}\}} &\mathbf{A}_{\{\text{I},\text{II}\}}  &\mathbf{0}  &\mathbf{A}_{\{\text{I},\text{II}\}} \\ 
        \mathbf{A}_{\{\text{II},\text{I}\}} &\mathbf{A}_{\{\text{II},\text{II}\}}  &\mathbf{0}  &(\mathbf{A}+\mathbf{D})_{\{\text{II},\text{II}\}}\\
        \mathbf{0} &\mathbf{0}  &-\mathbf{A}_{\{\text{I},\text{I}\}}  &-\mathbf{A}_{\{\text{I},\text{II}\}} \\ 
        \mathbf{0} &-\mathbf{D}_{\{\text{II},\text{II}\}}  &-\mathbf{A}_{\{\text{II},\text{I}\}}  &-(\mathbf{A}+\mathbf{D})_{\{\text{II},\text{II}\}}
    \end{pmatrix} \\[5pt]
    =&\begin{pmatrix}
        \mathbf{A}_{\{\text{I},\text{I}\}} &\mathbf{A}_{\{\text{I},\text{II}\}}  &\mathbf{0}  &\mathbf{0}\\ 
        \mathbf{A}_{\{\text{II},\text{I}\}} &\mathbf{A}_{\{\text{II},\text{II}\}}  &\mathbf{0}  &\mathbf{D}_{\{\text{II},\text{II}\}}\\
        \mathbf{0} &\mathbf{0}  &-\mathbf{A}_{\{\text{I},\text{I}\}}  &-\mathbf{A}_{\{\text{I},\text{II}\}} \\ 
        \mathbf{0} &-\mathbf{D}_{\{\text{II},\text{II}\}}  &-\mathbf{A}_{\{\text{II},\text{I}\}}  &-\mathbf{A}_{\{\text{II},\text{II}\}} 
    \end{pmatrix}=\mathbf{\Phi}^{12}\,.
\end{aligned}
\end{equation}

Note that $\text{I}\sqcup\text{II}=\{3,4,\dots,n\}$, the DR of prefactor is given by:
\begin{equation}
    (-\alpha')^{-q}\prod_{(ij)\in\rho}\frac{2 e_{i}\cdot e_{j}}{z_{i j}^2} \xrightarrow{\,\eqref{eq: DRgeneral}\,} (-1)^q \prod_{(ij)\in\rho}W_{ij}\times\begin{dcases}
    1 & \text{if}\  \rho\in(\mathfrak{S}_q|\rho\cap\text{I}\mspace{1mu})\\
    0 & \text{otherwise}
    \end{dcases}.
\end{equation}

Recall that we take derivative of $e_1\cdot e_2$ before DR, by using~\eqref{eq: expansion of Pf} we have:
\begin{equation}\label{eq: derivative of Pf Psi}
    \frac{\partial}{\partial e_1\cdot e_2}\operatorname{Pf}{\mathbf\Psi}_{12}^{\rho}=\frac{2}{z_{12}}\operatorname{Pf}{\mathbf\Psi}_{1,2,n+1,n+2}^{\rho}=\frac{2}{z_{12}}\operatorname{Pf}{\mathbf\Psi}^{(12)\mspace{1mu} \sqcup\mspace{1mu}\rho}.
\end{equation}

Now we combine the above results together to get:
\begin{equation}\label{eq: phi scalar to W}
\begin{aligned}
    \varphi^{\mathrm{scalar}}_{n} =\mathcal{DR}\left(\partial_{e_1\cdot e_2}\varphi_{n}^{\text{gauge}}\right) =\frac{2}{z_{12}^2}\sum_{q=0}^{\lfloor n/2\rfloor -1}\!(-1)^{q}\!\!\sum_{\rho\in(\mathfrak{S}_q|\rho\mspace{1mu}\cap\mspace{1mu}\text{I})}\prod_{(ij)\in\rho}W_{ij}\,\operatorname{Pf}{\mathbf\Phi}^{(12)\mspace{1mu} \sqcup\mspace{1mu}\rho}.
\end{aligned}
\end{equation}
where the somewhat self-referencing $\rho\in(\mathfrak{S}_q|\rho\mspace{1mu}\cap\mspace{1mu}\text{I})$ goes over all $\rho\in\mathfrak{S}_q$ such that each pair in $\rho$ has nonempty intersection with $\text{I}$, \textit{without} requiring that $\text{I}\subset\rho$. Later we will have this condition encoded in IBP building blocks to avoid this self-referencing notation.

The next key step of our proof is to expand $\operatorname{Pf}{\mathbf\Phi}^{(12)\mspace{1mu} \sqcup\mspace{1mu}\rho}$ into Taylor series of $V_i\,$s, where $\mathbf{\Phi}$ is regarded as function of $V_i\,$s through $(\mathbf{D})_{ij}$:
\begin{equation}\label{eq: Pf expansion in V}
    \operatorname{Pf}{\mathbf\Phi}^{(12)\mspace{1mu} \sqcup\mspace{1mu}\rho} =\sum_{s}\! \sum_{\{i_1,\ldots,i_s\}\subset \bar{\rho}\mspace{1mu}\cap\text{II}}\!\! V_{i_1}\mspace{-2mu}\cdots V_{i_s} \frac{\partial^r}{\partial V_{i_1}\cdots \partial V_{i_s}} \operatorname{Pf}{\mathbf\Phi}^{(12)\mspace{1mu} \sqcup\mspace{1mu}\rho} \left|_{V_k\rightarrow 0,\ \forall k\notin \{i_1,\ldots,i_s\}}\right. \,.
\end{equation}

Since $\operatorname{Pf}{\mathbf\Phi}^{(12)\mspace{1mu} \sqcup\mspace{1mu}\rho}$ contains at most linear terms with respect to each $V_i$, we omit all the higher-order terms in Taylor series. Let $\text{S}=\{i_1,\ldots,i_s\}$, by using~\eqref{eq: expansion of Pf} we can reduce~\eqref{eq: Pf expansion in V} into summation of $\text{det}\mspace{1mu}(\mathbf{A}_{(12)\mspace{1mu} \sqcup\mspace{1mu}\rho\mspace{1mu}\sqcup\mspace{1mu}\text{S}})$ over sets $\text{S}$:
\begin{equation}
    \operatorname{Pf}\mathbf{\Phi}^{(12)\mspace{1mu} \sqcup\mspace{1mu}\rho} =\sum_{\text{S}\subset\bar{\rho}\mspace{1mu}\cap\text{II}}(-1)^{\frac{(2n+3-\left |\text{S}\right |)\left |\text{S}\right |}{2}}(-1)^{\frac{n-|\rho|-2-|\text{S}|}{2}}\text{det}\mspace{1mu}(\mathbf{A}_{(12)\mspace{1mu} \sqcup\mspace{1mu}\rho\mspace{1mu}\sqcup\mspace{1mu}\text{S}}) \prod_{i\in\text{S}}V_{i}
\end{equation}

Note that the determinant of an odd-dimensional skew-matrix is zero, we further obtain:
\begin{equation}\label{eq: expand Pf phi^rd}
    \operatorname{Pf}\mathbf{\Phi}^{(12)\mspace{1mu} \sqcup\mspace{1mu}\rho}= (-1)^{\lfloor\frac{n-|\rho|-2}{2}\rfloor} \sum_{\substack{\text{S}\subset\bar{\rho}\mspace{1mu}\cap\text{II}\\ |\text{S}|+n=\text{even}}} \prod_{i\in\text{S}}V_{i}\; \text{det}\mspace{1mu}(\mathbf{A}_{(12)\mspace{1mu} \sqcup\mspace{1mu}\rho\mspace{1mu}\sqcup\mspace{1mu}\text{S}})
\end{equation}

Finally we collect~\eqref{eq: phi scalar to W} and~\eqref{eq: expand Pf phi^rd} together to get:
\begin{equation}\label{eq: phi to W V detA}
    \varphi^{\mathrm{scalar}}_{n} =\frac{2(-1)^{\lfloor\frac{n}{2}\rfloor-1}}{z_{12}^2} \sum_{q=0}^{\lfloor n/2\rfloor -1}\!\!\!\sum_{\substack{\rho\in(\mathfrak{S}_q|\rho\mspace{1mu}\cap\mspace{1mu}\text{I})\\[1pt] \text{S}\subset\bar{\rho}\mspace{1mu}\cap\text{II}\mspace{1mu},\mspace{1mu} |\text{S}|+n=\text{even}}}\!\! \prod_{(ij)\in\rho}W_{ij}\,\prod_{i\in\bar{\rho}}V_{i}\; \text{det}\mspace{1mu}(\mathbf{A}_{(12)\mspace{1mu} \sqcup\mspace{1mu}\rho\mspace{1mu}\sqcup\mspace{1mu}\text{S}}).
\end{equation}
where $\bar{\rho}$ is the complement of $\rho$ with respect to $\rho\mspace{1mu}\sqcup\mspace{1mu}\text{S}$. Let $\text{T}=\rho\mspace{1mu}\sqcup\mspace{1mu}\text{S}$, we have $\rho=\mathrm{T}\cap\mathrm{I}\,$, $\text{S}=\mathrm{T}\cap\mathrm{II}$. Note that now $\text{T}$ goes over all the subsets of $\text{I}\sqcup\text{II}$ with $\left |\mathrm{T}\right |+n=\mathrm{even}$, we can rewrite~\eqref{eq: phi to W V detA} as summation of IBP building blocks~\eqref{def: IBP building blocks} over all possible $\text{T}$ as \textit{sets}, and get rid of the somewhat annoying self-referencing notation:
\begin{equation}
    \varphi^{\mathrm{scalar}}_{n} =\frac{2(-1)^{\lfloor\frac{n}{2}\rfloor-1}}{z_{12}^2}\sum_{\substack{\mathrm{T}\subset\{3,\ldots,n\}\\ |\mathrm{T}|+n=\mathrm{even}}}\mathcal{I}(\mathrm{T}\cap\mathrm{I}\mspace{1mu}|\mspace{1mu}\mathrm{T}\cap\mathrm{II})\,\mathrm{det}\mspace{1mu}(\mathbf{A}_{(12)\mspace{1mu}\sqcup\mspace{1mu}\mathrm{T}}).
\end{equation}
which is exactly \myTheorem~\ref{thm: superstring linear expand} we expect. This completes the proof.

\section{Details for deriving the logarithmic form~\eqtitleref{eq: BCJnumerator} and explicit result}\label{app.: Supplement for mixed amplitude}
In this appendix we will provide the proofs of~\eqref{eq: gauge+color correlator} and~\eqref{eq: pion+color correlator}. Throughout this appendix we set $(i_0,j_0)=(1,n)$ in~\eqref{eq: superstring correlator} to be consistent with conventions in Section~\ref{sec: Mixed amplitudes and logarithmic forms} and~\cite{Cheung:2017ems}.
\subsection{Proof of~\eqtitleref{eq: gauge+color correlator}}
Recall that the definition of superstring YMS integrand reads:
\begin{equation*}
    \varphi^{\text{gauge}+\text{color}}_{n}(\{i_1,i_2,\ldots,i_m\}|1,\alpha,n)=\partial_{\mspace{2mu}2 e_{1} \cdot e_{n}} \prod_{i=1}^{|\alpha|} \partial_{\mspace{2mu}2 e_{\alpha_i} \cdot (p_{\alpha_{i-1}} - p_{n}) } \varphi^{\text{gauge}}_{n}, \tag{\ref{eq: gauge+color definition}}
\end{equation*}

In order to obtain an explicit formula of superstring YMS integrand from~\eqref{eq: gauge+color definition}, let us first take the derivative of $e_1\cdot e_n$, by using~\eqref{eq: expansion of Pf} we have:
\begin{equation}\label{eq: par e1.en of phi gauge}
    \partial_{\mspace{2mu}{2e_1\cdot e_n}}\varphi_{n}^{\text{gauge}}=\frac{1}{z_{1n}}\sum_{q=0}^{\lfloor n/2\rfloor -1}(-\alpha')^{-q}\sum_{\rho\in\mathfrak{S}_q}\prod_{(ij)\in\rho}\frac{2 e_{i}\cdot e_{j}}{z_{i j}^2}\;\operatorname{Pf}\mathbf{\Psi}^{1,\rho,n}
\end{equation}

Note that for any $\rho$ such that $\rho\cap\alpha\neq\varnothing$, the corresponding term in~\eqref{eq: par e1.en of phi gauge} does not contribute to the final result, since $\partial_{\mspace{2mu}\raisemath{3pt}{2e_{\alpha_i} \cdot (p_{\alpha_{i-1}} - p_{n})}}$ yields zero for both $\prod_{(ij)\in\rho} 2e_i\cdot e_j$ and $\operatorname{Pf}\mathbf{\Psi}^{1,\rho,n}$ for any $\alpha_i\in\rho\cap\alpha$. Thus we can safely restrict the range of gauge particles (where elements of $\rho\in\mathfrak{S}_q$ is taken) from $\{2,\dots,n-1\}$ to $\{i_1,\dots,i_m\}$. In the following text, elements of partitions in $\mathfrak{S}_q$ is always taken from $\{i_1,\dots,i_m\}$.


The next task is to evaluate the derivatives $\partial_{\mspace{2mu}\raisemath{3pt}{2e_{\alpha_i} \cdot\, (p_{\alpha_{i-1}}-p_n)}}$. Since $\alpha_0=1\in(1,\rho,n)$, the only component of $\mathbf{\Psi}^{1,\rho,n}$ containing $e_{\alpha_1}\!\cdot p_n$ or $e_{\alpha_1}\!\cdot p_{\alpha_0}$ is $(\mathbf{C})_{\alpha_{1}\alpha_{1}}$, thus the derivative $\partial_{\mspace{2mu}\raisemath{3pt}{2e_{\alpha_{1}} \cdot\, (p_{\alpha_0}-p_n)}}$ is equivalent to $(\partial_{\mspace{2mu}\raisemath{3pt}{2e_{\alpha_{1}} \cdot\, (p_{\alpha_0}-p_n)}} (\mathbf{C})_{\alpha_{1}\alpha_{1}})\,\partial_{\raisemath{2pt}{(\mathbf{C})_{\alpha_{1}\alpha_{1}}}}$ when acting on $\operatorname{Pf}\mathbf{\Psi}^{1,\rho,n}$. This inspire us to translate all the derivatives $\partial_{\mspace{2mu}\raisemath{3pt}{2e_{\alpha_i} \cdot\, (p_{\alpha_{i-1}}-p_n)}}$ into $\partial_{\raisemath{2pt}{(\mathbf{C})_{\alpha_{i}\alpha_{i}}}}$, which reads:
\begin{equation}\label{eq: gauge+color translation}
    \prod_{i=1}^{|\alpha|} \partial_{\mspace{2mu} 2e_{\alpha_i} \cdot (p_{\alpha_{i-1}}-p_n)} \cong \Bigg(\prod_{i=1}^{|\alpha|} \partial_{\mspace{2mu} 2e_{\alpha_i} \cdot (p_{\alpha_{i-1}}-p_n)} (\mathbf{C})_{\alpha_{i}\alpha_{i}}\Bigg)\prod_{i=1}^{|\alpha|} \partial_{(\mathbf{C})_{\alpha_{i}\alpha_{i}}}
\end{equation}


This translation holds if the input of $\partial_{\mspace{2mu}\raisemath{3pt}{2e_{\alpha_i} \cdot\, (p_{\alpha_{i-1}}-p_n)}}$ depends on $e_{\alpha_i}\!\cdot p_n$ and $e_{\alpha_i}\!\cdot p_{\alpha_{i-1}}$ solely through $(\mathbf{C})_{\alpha_{i}\alpha_{i}}$, and it is the case. By recursively using~\eqref{eq: expansion of Pf} we have:
\begin{equation}
    \prod_{i=1}^{k} \partial_{(\mathbf{C})_{\alpha_{i}\alpha_{i}}}\operatorname{Pf}\mathbf{\Psi}^{1,\rho,n}=(-1)^{\frac{k(2n-3-k)}{2}}\operatorname{Pf}\mathbf{\Psi}^{1,\alpha_1,\dots,\alpha_k,\rho,n},\quad \forall k\leqslant |\alpha|.
\end{equation}

Thus the input of $\partial_{\mspace{2mu}\raisemath{3pt}{2e_{\alpha_i} \cdot\, (p_{\alpha_{i-1}}-p_n)}}$ is just $\operatorname{Pf}\mathbf{\Psi}^{1,\alpha_1,\dots,\alpha_{i-1},\rho,n}$ (up to an overall sign and a prefactor), whose only component containing $e_{\alpha_i} \!\cdot p_n$ or $e_{\alpha_i}\!\cdot p_{\alpha_{i-1}}$ is $(\mathbf{C})_{\alpha_{i}\alpha_{i}}$. This proves the translation~\eqref{eq: gauge+color translation}. The prefactor is easy to get:
\begin{equation}\label{eq: gauge+color prefactor}
    \prod_{i=1}^{|\alpha|} \partial_{\mspace{2mu}2e_{\alpha_i} \cdot (p_{\alpha_{i-1}}-p_n)} (\mathbf{C})_{\alpha_{i}\alpha_{i}}=\prod_{i=1}^{|\alpha|} \frac{z_{\alpha_{i-1} n}}{z_{\alpha_{i-1} \alpha_{i}} z_{\alpha_{i} n}}=\frac{\operatorname{PT}(1,\alpha,n)}{\operatorname{PT}(1,n)}
\end{equation}


Finally we collect~\eqref{eq: gauge+color translation} to~\eqref{eq: gauge+color prefactor} together, and take $k=|\alpha|$ to get:
\begin{equation}
\begin{aligned}
    \varphi^{\text{gauge}+\text{color}}_{n} & (\{i_1,\ldots,i_m\} |1,\alpha,n)=(-1)^{n+1+\frac{|\alpha|(2n-3-|\alpha|)}{2}}\times\\ \times\,&\mathrm{PT}(1,\alpha,n)\sum_{q=0}^{\lfloor m/2 \rfloor} (-\alpha')^{-q} \sum_{\rho\in\mathfrak{S}_q} \prod_{(ij)\in\rho} \frac{2 e_{i}\mspace{-1mu}\cdot\mspace{-1mu} e_{j}}{z_{i j}^2}\,\operatorname{Pf}\mathbf{\Psi}^{1,\alpha,\rho,n}.
\end{aligned}
\end{equation}

Neglecting the unimportant overall sign, this is exactly~\eqref{eq: gauge+color correlator} we desire.

\subsection{Proof of~\eqtitleref{eq: pion+color correlator}}
In this subsection, we fix $z_n\to \infty$, and denote the identities that only hold after omitting $z_{in}\to \infty$ by ``$\doteq$''. These $z_{in}$ can be safely omitted since they automatically cancel out with the $z_{jn}\to \infty$ in $d\mu_n^\text{string}$. In order to prove~\eqref{eq: pion+color correlator}, let us first expand~\eqref{eq: mixed amplitude in det(A) PT} into summation of labelled trees for the case that $\alpha=\varnothing$ with the matrix tree theorem~\cite{Feng:2012sy}, then reduce to general $\alpha$ with Laplace expansion of determinant.

For $\alpha=\varnothing$, the matrix tree theorem can be directly applied to yield the result:
\begin{equation}\label{eq: matrix tree for pure pion}
	\operatorname{det}\mathbf{\tilde A}_{\{2,3,\dots,n-1\}}\operatorname{PT}(1,n) \doteq (-1)^n\sum_{G(1,2,\dots,n-1)} \prod_{e(i,j)} \frac{s_{i j}}{z_{i j}},
\end{equation}
where we make the denominator of each term matching the form $z_{2,\bullet}z_{3,\bullet}\cdots z_{n-1,\bullet}$ to get the correct relative signs. For instance:
\begin{equation}
	\operatorname{det}\mathbf{\tilde A}_{\{2,3\}}\operatorname{PT}(1,4) \doteq \frac{s_{21} s_{31}}{z_{21} z_{31}}+\frac{s_{23} s_{31}}{z_{23} z_{31}}+\frac{s_{21} s_{32}}{z_{21} z_{32}}.
\end{equation}

By treating $s_{ij}|_{1\leqslant i<j\leqslant n}$ as independent variables, similar arguments as the proof of~\eqref{eq: gauge+color correlator} yields the following equivalence for~\eqref{eq: matrix tree for pure pion}:
\begin{equation}\label{eq: matrix tree translation}
	\prod_{i=1}^{|\alpha|} \partial_{s_{\alpha_{i-1} \alpha_{i}}} \cong \Bigg(\prod_{i=1}^{|\alpha|} \partial_{s_{\alpha_{i-1} \alpha_{i}}}(\mathbf{\tilde A})_{\alpha_{i}\alpha_{i}}\Bigg)\prod_{i=1}^{|\alpha|} \partial_{(\mathbf{\tilde A})_{\alpha_{i}\alpha_{i}}},
\end{equation}

Then we can apply Laplace expansion to recursively remove pions from $\operatorname{det}\mathbf{\tilde A}_{\{2,3,\dots,n-1\}}$:
\begin{equation}
    \partial_{(\mathbf{\tilde A})_{\alpha_{i}\alpha_{i}}}\operatorname{det}\mathbf{\tilde A}_{\{\ldots,\alpha_{i-1},\alpha_{i},\alpha_{i+1},\ldots \}} =\operatorname{det}\mathbf{\tilde A}_{\{\ldots,\alpha_{i-1},\alpha_{i+1},\ldots \}}.
\end{equation}

And the prefactor evaluates to:
\begin{equation}\label{eq: matrix tree prefactor}
    \prod_{i=1}^{|\alpha|} \partial_{s_{\alpha_{i-1} \alpha_{i}}} (\mathbf{\tilde A})_{\alpha_{i}\alpha_{i}}=\prod_{i=1}^{|\alpha|} \frac{1}{z_{\alpha_{i-1} \alpha_{i}}} \doteq \operatorname{PT}(1,\alpha,n).
\end{equation}

Finally, we collect~\eqref{eq: matrix tree translation} to~\eqref{eq: matrix tree prefactor} together to get:
\begin{equation}
	 \operatorname{det}\mathbf{\tilde A}_{\{i_1,i_2,\ldots,i_m \}}\operatorname{PT}(1,\alpha,n)=\prod_{i=1}^{|\alpha|}\partial_{s_{\alpha_{i-1}\alpha_{i}}}\operatorname{det}\mathbf{\tilde A}_{\{2,3,\dots,n-1\}}.
\end{equation}

The RHS selects all the trees containing sub-tree $(1,\alpha)$ divided by $\prod_{i=1}^{|\alpha|} s_{\alpha_{i-1} \alpha_{i}}$:
\begin{equation}
    \operatorname{det}\mathbf{\tilde A}_{\{i_1,i_2,\ldots,i_m \}}\operatorname{PT}(1,\alpha,n) \doteq (-1)^{n}\sum_{G(1,\alpha)} \prod_{e(i,j)} \frac{s_{i j}}{z_{i j}} \prod_{k=1}^{|\alpha|} \frac{1}{s_{\alpha_{k-1},\alpha_{k}} }.
\end{equation}

This is exactly~\eqref{eq: pion+color correlator}, hence completes the proof.

\subsection{Explicit example at 9 points}
The field theory 9-point 6 pions amplitude can be computed via similar formula as~\eqref{eq: BCJnumerator}, where we need to take {\it e.g.} $(i_0,j_0)=(3,9)$ to be special instead of $(1,9)$, then the result is obtained by plugging in the bi-adjoint $\phi^3$ amplitudes:
\begin{equation*}
\resizebox{\textwidth}{!}{$\begin{aligned}
&A_{9}^{\mathrm{NLSM}+\phi ^3}(\{1,2,4,5,7,8\}|3,6,9)=-\frac{5}{3}+\frac{2\left( X_{1,3}+X_{2,4} \right)}{X_{1,4}}+\frac{2\left( X_{1,8}+X_{2,9} \right)}{X_{2,8}}+\frac{2\left( X_{1,3}+X_{2,9} \right)}{X_{3,9}}\\
&+\frac{\left( X_{1,3}+X_{2,4} \right) \left( X_{4,6}+X_{5,7} \right) \left( X_{1,8}+X_{7,9} \right)}{3X_{1,4}X_{1,7}X_{4,7}}+\frac{\left( X_{1,8}+X_{2,9} \right) \left( X_{2,4}+X_{3,5} \right) \left( X_{5,7}+X_{6,8} \right)}{3X_{2,5}X_{2,8}X_{5,8}}\\
&+\frac{\left( X_{1,3}+X_{2,9} \right) \left( X_{3,5}+X_{4,6} \right) \left( X_{6,8}+X_{7,9} \right)}{3X_{3,6}X_{3,9}X_{6,9}}+\frac{\left( X_{1,8}+X_{2,9} \right) \left( X_{2,7}+X_{3,8} \right) \left( X_{3,5}+X_{4,6} \right)}{X_{2,8}X_{3,6}X_{3,7}}\\
&+\frac{\left( X_{1,8}+X_{2,9} \right) \left( X_{2,7}+X_{3,8} \right) \left( X_{4,6}+X_{5,7} \right)}{X_{2,8}X_{3,7}X_{4,7}}+\frac{\left( X_{1,3}+X_{2,4} \right) \left( X_{1,5}+X_{4,6} \right) \left( X_{1,8}+X_{7,9} \right)}{X_{1,4}X_{1,6}X_{1,7}}\\
&+\frac{\left( X_{1,3}+X_{2,6} \right) \left( X_{3,5}+X_{4,6} \right) \left( X_{1,8}+X_{7,9} \right)}{X_{1,6}X_{1,7}X_{3,6}}+\frac{\left( X_{1,3}+X_{2,7} \right) \left( X_{3,5}+X_{4,6} \right) \left( X_{1,8}+X_{7,9} \right)}{X_{1,7}X_{3,6}X_{3,7}}\\
&+\frac{\left( X_{1,3}+X_{2,9} \right) \left( X_{3,5}+X_{4,6} \right) \left( X_{3,8}+X_{7,9} \right)}{X_{3,6}X_{3,7}X_{3,9}}+\frac{X_{1,3}+X_{1,5}+X_{2,4}+X_{2,6}+X_{3,5}+X_{4,6}}{X_{1,6}}\\
&-\frac{\left( X_{1,5}+X_{2,6} \right) \left( X_{2,4}+X_{3,5} \right)}{X_{1,6}X_{2,5}}-\frac{\left( X_{1,8}+X_{2,9} \right) \left( X_{2,4}+X_{3,5} \right)}{X_{2,5}X_{2,8}}-\frac{\left( X_{1,3}+X_{2,4} \right) \left( X_{1,5}+X_{4,6} \right)}{X_{1,4}X_{1,6}}\\
&-\frac{\left( X_{1,3}+X_{2,6} \right) \left( X_{3,5}+X_{4,6} \right)}{X_{1,6}X_{3,6}}-\frac{\left( X_{1,8}+X_{2,9} \right) \left( X_{3,5}+X_{4,6} \right)}{X_{2,8}X_{3,6}}-\frac{\left( X_{1,3}+X_{2,9} \right) \left( X_{3,5}+X_{4,6} \right)}{X_{3,6}X_{3,9}}\\
&-\frac{\left( X_{1,3}+X_{2,4} \right) \left( X_{4,6}+X_{5,7} \right)}{X_{1,4}X_{4,7}}-\frac{\left( X_{1,3}+X_{2,9} \right) \left( X_{4,6}+X_{5,7} \right)}{X_{3,9}X_{4,7}}-\frac{\left( X_{1,8}+X_{2,9} \right) \left( X_{4,6}+X_{5,7} \right)}{X_{2,8}X_{4,7}}\\
&-\frac{\left( X_{1,3}+X_{1,5}+X_{2,4}+X_{2,6}+X_{3,5}+X_{4,6} \right) \left( X_{1,8}+X_{7,9} \right)}{X_{1,6}X_{1,7}}\\
&-\frac{\left( X_{3,5}+X_{4,6} \right) \left( X_{1,3}+X_{1,8}+X_{2,7}+X_{2,9}+X_{3,8}+X_{7,9} \right)}{X_{3,6}X_{3,7}}+(\mathrm{cyclic},i\rightarrow i+3,i+6).
\end{aligned}$}
\end{equation*}

\bibliographystyle{JHEP}\bibliography{Refs}

\end{document}